%% file: main.tex
\documentclass[format=acmsmall, review=false]{acmart}
\usepackage{acm-ec-24}
\usepackage{booktabs} %
\usepackage[ruled,linesnumbered]{algorithm2e} %
\usepackage{natbib} %

\SetAlFnt{\small}
\SetAlCapFnt{\small}
\SetAlCapNameFnt{\small}
\SetAlCapHSkip{0pt}
\IncMargin{-\parindent}

\setcitestyle{acmnumeric}

\title{Sample-Based Matroid Prophet Inequalities}

\author{Hu Fu}
\email{fuhu@mail.shufe.edu.cn}
\orcid{0009-0005-4217-4329}
\affiliation{%
  \department{ITCS, Key Laboratory of Interdisciplinary Research of Computation and Economics}
  \institution{Shanghai University of Finance and Economics}
  \country{China}
}

\author{Pinyan Lu}
\email{lu.pinyan@mail.shufe.edu.cn}
\orcid{0009-0005-0569-4122}
\affiliation{%
  \department{ITCS, Key Laboratory of Interdisciplinary Research of Computation and Economics}
  \institution{Shanghai University of Finance and Economics}
  \country{China}
}

\author{Zhihao Gavin Tang}
\email{tang.zhihao@mail.shufe.edu.cn}
\orcid{0000-0002-5094-1971}
\affiliation{%
  \department{ITCS, Key Laboratory of Interdisciplinary Research of Computation and Economics}
  \institution{Shanghai University of Finance and Economics}
  \country{China}
}

\author{Hongxun Wu}
\email{wuhx@berkeley.edu}
\orcid{0009-0005-5544-7517}
\affiliation{%
  \institution{UC Berkeley}
  \country{USA}
}

\author{Jinzhao Wu}
\email{jinzhao.wu@yale.edu}
\orcid{0000-0003-3068-7475}
\affiliation{%
  \institution{Yale University}
  \country{USA}
}

\author{Qianfan Zhang}
\email{qianfan@princeton.edu}
\orcid{0000-0003-3737-1545}
\affiliation{%
  \institution{Princeton University}
  \country{USA}
}

\usepackage{enumitem}
\usepackage{microtype}
\usepackage{graphicx}
\usepackage{subcaption}

\usepackage{amsmath, amsthm, amssymb}
\usepackage{bbm}
\usepackage{mathtools}
\usepackage{thm-restate}
\usepackage{natbib}
\usepackage{color}
\usepackage{xcolor}

\usepackage{float}
\usepackage{libertine}
\usepackage{tikz,pgfmath}
\usetikzlibrary{matrix,positioning,quotes}
\usetikzlibrary{shapes.multipart}
\usetikzlibrary{calc}
\usetikzlibrary{matrix}

\usepackage[title]{appendix}

\usepackage{cleveref}
\usepackage{xfrac}

\input{notation}

\theoremstyle{plain}
\newtheorem{theorem}{Theorem}[section]
\newtheorem{lemma}[theorem]{Lemma}

\newtheorem{fact}[theorem]{Fact}
\newtheorem{observation}[theorem]{Observation}

\newtheorem{remark}[theorem]{Remark}

\begin{abstract}
We study matroid prophet inequalities when distributions are unknown and accessible only through samples. While single-sample prophet inequalities for special matroids are known, no constant-factor competitive algorithm with even a sublinear number of samples was known for general matroids. Adding more to the stake, the single-sample version of the question for general matroids has close (two-way) connections with the long-standing matroid secretary conjecture.

In this work, we give a \((\sfrac14 - \varepsilon)\)-competitive matroid prophet inequality with only $O_\varepsilon(\poly \log n)$ samples. Our algorithm consists of two parts: (i) a novel \emph{quantile-based} reduction from matroid prophet inequalities to online contention resolution schemes (OCRSs) with \(O_\varepsilon(\log n)\) samples, and (ii) a \((\sfrac14 - \varepsilon)\)-selectable matroid OCRS with \(O_\varepsilon(\poly \log n)\) samples which carefully addresses an adaptivity challenge.

\end{abstract}

\begin{document}

\maketitle
\newpage

\section{Introduction}
\label{sec:intro}
\input{intro}

\section{Preliminaries}
\input{prelim}

\section{Matroid Prophet Inequalities from Samples}
\label{sec:prophet}
\input{matroid_prophet}

\section{Matroid OCRS from Samples}
\label{sec:orcs}
\input{matroid_ocrs}

\bibliographystyle{ACM-Reference-Format}
\bibliography{bibs}

\input{appendix}

\end{document}

%% file: notation.tex
\newcommand{\AutoAdjust}[3]{\mathchoice{ \left #1 #2  \right #3}{#1 #2 #3}{#1 #2 #3}{#1 #2 #3} }
\newcommand{\Xcomment}[1]{{}}

\newcommand{\InParentheses}[1]{\AutoAdjust{(}{#1}{)}}
\newcommand{\InBrackets}[1]{\AutoAdjust{[}{#1}{]}}%
\newcommand{\Ex}[2][]{\operatorname{\mathbf E}_{#1}\InBrackets{#2}}
\newcommand{\E}[2][]{\operatorname{\mathbf E}_{#1}\InBrackets{#2}}
\newcommand{\Prx}[2][]{\operatorname{\mathbf{Pr}}_{#1}\InBrackets{#2}}
\renewcommand{\Pr}{\mathbf{Pr}}

\newcommand{\dd}{\mathrm{d}}  %
\newcommand{\given}{\;\vert\;}

\newcommand{\eqdef}{\stackrel{\textnormal{def}}{=}}

\newcommand{\noaccents}[1]{#1}

\newcommand{\cF}{\mathcal{F}}

\newcommand{\newagentvar}[3][\noaccents]{%
\expandafter\newcommand\expandafter{\csname #2\endcsname}{#1{#3}}%
\expandafter\newcommand\expandafter{\csname #2s\endcsname}{#1{\boldsymbol{#3}}}%
\expandafter\newcommand\expandafter{\csname #2smi\endcsname}[1][i]{#1{\boldsymbol{#3}}_{\text{-}##1}}%
\expandafter\newcommand\expandafter{\csname #2i\endcsname}[1][i]{#1{#3}_{##1}}%
\expandafter\newcommand\expandafter{\csname #2ith\endcsname}[1][i]{#1{#3}_{(##1)}}%
}
\DeclareMathOperator{\poly}{poly}

\DeclareMathOperator{\opt}{OPT}
\DeclareMathOperator{\OPT}{OPT}
\DeclareMathOperator{\ALG}{ALG}

\DeclareMathOperator{\argmax}{argmax}

\DeclareMathOperator{\Span}{span}

\DeclareMathOperator{\rank}{rank}

\newcommand{\matroid}{\mathcal M}
\newcommand{\indset}{\mathcal I}

\newagentvar{val}{v}
\newagentvar[\tilde]{mulval}{v}
\newagentvar[\tilde]{multval}{v}
\newagentvar[\hat]{divval}{v}
\newagentvar{dist}{\mathcal{D}}
\newagentvar{disttau}{G}

\newagentvar{thresh}{\theta}

\newcommand{\event}{\mathcal{E}}

\newcommand{\polytope}{\mathcal{P}}
\newagentvar{alloc}{x}
\newagentvar{avgalloc}{y}
\newagentvar{optalloc}{y}

\newcommand{\I}[2][]{\operatorname{\mathbbm 1}_{#1}\InBrackets{#2}}

%% file: intro.tex
In the classical prophet inequality problem introduced by \citet{KS77,KS78}, $n$ items, with values $\val_1,\val_2,\dots,\val_n$ drawn independently from known distributions $\dist_1,\dist_2,\dots,\dist_n$, arrive one by one.
An algorithm, upon observing the value of each item, must decide irrevocably whether to take it or discard it forever.
The algorithm can take at most one item in total, and aims to maximize the expected value of the item taken.
we measure the performance of an algorithm via its \emph{competitive ratio}: an algorithm is \emph{$c$-competitive} (for $c \in [0, 1]$) if its expected value is at least $c$ fraction of what can be achieved by a \emph{prophet}, who knows the realizations of the values and always takes the highest value $\max_i \vali$, getting in expectation $\Ex{\max_i \vali}$.
Krengel, Sucheston, and Garling~\citep{KS77,KS78} showed that $1/2$ is the optimal competitive ratio in this so-called single-item setting.

Beyond interest of their own as canonical optimal stopping problems, prophet inequalities also have strong connections with online mechanism design and posted price mechanisms~\citep{HKS07,CHMS10}. 
Many variants beyond the single-item case have been studied extensively (e.g., \citep{KW12,Alaei14,FGL14,FSZ16,Rub16,RS17,DKL20,DFKL20,ehsani2018prophet,EFGT20,hill1982comparisons,correa2017posted,peng2022order,alaei2012online,GW19,correa2023constant}). 
One of the significant extensions is \emph{matroid prophet inequalities} discovered by~\citet{KW12}, where 
the set of items taken is subject to a matroid feasibility constraint.
The single-item setting is the special case when the matroid is the 1-uniform matroid.
\citet{KW12} presented a $1/2$-competitive algorithm that works for any matroid, which matches the best possible ratio even in the single-item case.

Many existing works on prophet inequalities assumed complete knowledge of distributions, which may be unrealistic both in practice and for some applications such as prior-independent mechanism design~\citep{HR09,DRY10, azar2019prior}. It is natural to ask whether similar performance can be achieved with only limited access to distributions (e.g., historical data points). This question was first formulated as \emph{prophet inequalities with limited information} (specifically, \emph{sample-based} prophet inequalities) by~\citet{AKW14}, where the online algorithm only has access to a limited number of samples from each distribution $\disti$. Such a sample-based paradigm is closer to practice and naturally provides some robustness against distribution shifts. 

More specifically, there are two particularly interesting facets of the problem:
\begin{enumerate}[label={\arabic*.}]
     \item \emph{What is the best competitive ratio achievable with a single/constant number of samples?}
    \item \emph{What is the least number of samples needed to achieve a constant competitive ratio?}
\end{enumerate}
   
The first facet was studied by~\citet{AKW14}. They showed that a single sample from each distribution suffices for constant-factor competitive prophet inequalities for a number of feasibility constraints.
Since then, there have been many subsequent works on different settings such as single-item~\citep{RWW20}, identical distributions~\citep{CDFS19,KNR20,CDFSZ21,CCES23}, non-adversarial arrivals~\citep{CCES20,CZ23}, and various other feasibility constraints~\citep{KNR22,CDFFLLPPR22}. 
While some of these works studied sample-based prophet inequalities for specific classes of matroids, such as graphic or transversal matroids, not much was known for general matroids.
Notably, \citet{AKW14} design single-sample prophet inequalities via a black-box reduction to a large subclass of algorithms for \emph{secretary problems}\footnote{The secretary problem~\citep{Dyn63} is another classical model in online decision making that is closely related to prophet inequalities, where values are chosen adversarially rather than sampled from distributions, but items arrive in a uniformly random order.}. Based on the reduction, they obtained a single-sample $1/O(\log \log \rank)$-competitive prophet inequality for general matroids from the state-of-the-art matroid secretary algorithm~\citep{Lac14,FSZ14} with the same ratio.

It turns out that the sample-based prophet inequalities and the secretary problem are deeply connected. In fact, solving either facets completely requires resolving the famous matroid secretary conjecture by~\citet{BIK07}. For example, \citet{CDFFLLPPR22} discovered a partial reverse to the aforementioned reduction of \citet{AKW14}.
Very recently, \citet{Wen23} obtained a black-box reduction from the secretary problem to single-sample prophet inequalities (which also extends to the constant-sample prophet inequalities).
As a consequence, simultaneously achieving constant competitive ratio and constant number of samples implies a constant-factor competitive algorithm for matroid secretary problem, which would resolve the matroid secretary conjecture.

In this work, we make progress on the second facet of the problem. We show that $O_{\varepsilon}(\log^4 n)$ samples suffice for $(1 / 4 - \varepsilon)$-competitive prophet inequalities for general matroids. Before our work, with even $o(n)$ samples, no constant-factor competitive prophet inequalities for general matroids were known. This is due to the limitations of previous techniques:
\begin{itemize}
    \item In the original algorithm of~\citet{KW12}, upon the arrival of each item, its value is compared against a threshold. Crucially, this threshold is the expectation of a random variable that depends on the items currently taken. 
    \begin{itemize}
        \item There is the \emph{lack-of-concentration issue} with this strategy: Since the values of the items are unbounded, there is no concentration with any bounded number of samples. This prevents us from estimating the expectation threshold to any non-trivial accuracy.
        \item Even if the values are bounded and we have concentration, there is still the \emph{adaptivity issue}: The random variable for the item $i$ depends on whether items $1, 2, \dots, i - 1$ are accepted and, therefore, depends on previous estimates. Even when there is concentration, one still cannot union bound over such adaptive estimations. 
    \end{itemize}
\item The approach of~\citet{FSZ16} reduces prophet inequalities to online contention resolution schemes (OCRSs). The reduction also needs to set similar thresholds.
    \begin{itemize}
        \item It partially gets around the \emph{lack-of-concentration issue} by directly using the realization of some random variable as the threshold (rather than the expectation). To generate one such realization, they need one sample from each of the distributions $\dist_1,\dist_2,\dots,\dist_n$. Due to the requirement of OCRSs, they must use independent realizations for the $n$ items, this already requires $\Omega(n)$ samples. 
        \item Although there is no adaptivity in the reduction, their OCRSs also require full knowledge of the distributions. If one tries to implement the OCRSs with samples, a similar \emph{adaptivity issue} shows up. 
    \end{itemize}
\end{itemize}

We overcome these two issues with new techniques. Our first technical contribution is a quantile-based strategy for setting thresholds (in the reduction to OCRSs) and its analysis. Roughly speaking, we use the \emph{median} of some random variables as thresholds, which is reminiscent of the classical \citet{SC84} result for the rank-1 case. With a few samples, the estimate of the median is guaranteed to be a $(1/2 - \varepsilon, 1/2 + \varepsilon)$ quantile for some $\varepsilon > 0$ with high probability. This overcomes the \emph{lack-of-concentration issue}. The analysis of this strategy is highly nontrivial and involves a weighted generalization of matroid strong basis exchange lemma from submodular optimization literature. To the best of our knowledge, it has not been applied to matroid prophet inequalities before. More details are discussed in \Cref{sec:prophet-overview}. 

Our second contribution is a sample-based matroid OCRS overcoming the \emph{adaptivity issue}, which is an interesting result on its own. To obtain such OCRS, we necessarily need to randomize the matroid OCRS of~\citet{FSZ16} in an ingenious way to handle adaptivity. The approach we use is specific to OCRSs and differs from the usual randomization methods in adaptive data analysis. More details are discussed in \Cref{sec:orcs}. In conclusion, we believe that this work brings interesting techniques to the study of matroid prophet inequalities with limited information. 

\subsection{Further Related Works}

\citet{FSZ16} introduce the concept of Online Contention Resolution Schemes (OCRSs) and present a $\frac 1 4$-selectable OCRS for matroids. Our sampled-based OCRS in Section~\ref{sec:orcs} is based on their algorithm. This algorithm is later improved by \citet{LS18} who obtain a $\frac 1 2$-selectable matroid OCRS by a clever reduction to the matroid prophet inequalities problem. OCRSs are also considered in other settings including matching \cite{EFGT20}, knapsack problem \cite{JMZ22}, and $k$-uniform matroids \cite{JMZ22, DW23}. Recent studies have expanded to investigate OCRSs with limited information. \citet{HPT22} show that there exists a \(\frac{1}{e}\)-selectable oblivious OCRS for rank-1 matroids. In contrast, for graphical and transversal matroids, they prove that no OCRS utilizing a constant number of samples can attain a constant selectability.  \citet{santiago2023simple} address a scenario in which the marginal active probability $\alloc_e$ is disclosed upon the arrival of element $e$. They provide a constant selectable random-order OCRS for graphical matroids.

Online decision making with limited information is attracting growing attention. In addition to the sample-based prophet inequalities discussed earlier, \citet{li2023prophet} explore  I.I.D. prophet inequalities with only  access to  quantile queries and \citet{nuti2023secretary} study maximizing the probability of selecting the largest number in the secretary problem with a single sample. There are some works on posted-price mechanisms that also imply prophet inequalities with polynomially many samples \cite{DFKL20,banihashem2024power}.

%% file: prelim.tex
Throughout the paper, we use bold letters $\vals, \allocs$ to denote vectors and $\vali, \alloci$ to denote their entries.

\paragraph{Prophet Inequalities.}
In the prophet inequality problem, we are given a ground set of items~$U=[n]$, a downward-closed family of feasible sets $\cF \subseteq 2^{U}$, and a distribution $\disti$ for each $i \in U$.
Items arrive in a fixed order $1,2,\dots,n$.\footnote{Our sample-based algorithm for matroid prophet inequality/OCRS actually works against \emph{almighty adversary}, who determines the arrival order adaptively with full knowledge of all realizations of randomness and the decisions made by the algorithm (see Remark~\ref{rem:almighty-adversary} for details). Nevertheless, we assume a fixed order (i.e., against \emph{offline adversary}) throughout the paper for ease of reading. \label{footnote:arrival-order}}
As each item arrives, a value $\vali$ independently drawn from~$\disti$ is revealed to the algorithm, and an irrevocable decision must be made whether to include the item in the output $A$, while keeping $A \in \cF$.

For $c \in [0,1]$, we say the algorithm induces a {\em $c$-competitive} prophet inequality for $\cF$ if, for any $\dist_1,\dist_2,\dots,\dist_n$,
$$\Ex{\sum_{i \in A} \vali} \ge c\Ex{\max_{S \in \cF} \sum_{i \in S} \vali}$$
where the expectation is taken with respect to the joint distribution of $\dist_1,\dist_2,\dots,\dist_n$ and internal randomness of the algorithm (including samples in the limited information setting).

In the limited information setting, instead of distributions $\dist_1,\dist_2,\dots,\dist_n$, the algorithm will take a number of samples from each $\disti$ as the offline input.
We say an algorithm uses $k$ samples if it takes $k$ samples from each $\disti$.

\paragraph{Online Contention Resolution Schemes (OCRSs)}
Given a ground set of elements~$U=[n]$ and a downward-closed family of feasible sets $\cF \subseteq 2^{U}$, we can define the polytope of $\cF$ as the convex hull of all characteristic vectors of feasible sets:
$$\polytope_\cF = \mathrm{conv}(\{\mathbf{1}_I \mid I \in \cF\}) \subseteq [0,1]^n.$$

An OCRS for $\polytope_\cF$ takes an vector $\allocs \in \polytope_\cF$ as input.
Let $R(\allocs) \subseteq U$ be a random set where each element $i \in U$ is in $R(\allocs)$ with probability $\alloci$ independently.
The algorithm sees the membership in~$R(\allocs)$ of the elements in~$U$ arriving in a fixed order $1,2,\dots,n$;\footref{footnote:arrival-order} when each element $i \in U$ arrives, if it is active (i.e., $i \in R(\allocs)$), the OCRS must decide irrevocably whether to include $i$ in its output~$A$ while keeping $A \in \cF$.

For $c \in [0,1]$, an OCRS is said to be \emph{$c$-selectable} for $\cF$ if, for any $\allocs \in \polytope_{\cF}$,
$$\Pr[i \in A \mid i \in R(\allocs)] \ge c \quad \forall i \in U$$
where the probability is measured with respect to $R(\allocs)$ and internal randomness of the OCRS (including samples in the limited information setting).

In the limited information setting, instead of $\allocs$, the algorithm will take a number of samples of $R(\allocs)$ as the offline input.
We say an OCRS uses $k$ samples if it takes $k$ samples of $R(\allocs)$.

\paragraph{Matroids}

A matroid $\matroid = (U,\indset)$ is defined by a ground set of elements~$U=[n]$ and a non-empty downward-closed family of {\em independent sets} $\cF \subseteq 2^{U}$ with the {\em exchange property}, i.e., for every $A,B \in \indset$ where $|A|>|B|$, there exists some $i \in A \setminus B$ such that $B \cup \{i\} \in \indset$.
Some notations for matroids are used throughout the paper:
\begin{itemize}
    \item The {\em rank} of a set $S \subseteq U$ is the size of the largest independent set contained in $S$: 
    $$\rank(S)=\max\{|I| \mid I \in \indset, I \subseteq S\}.$$
    \item The {\em span} of a set $S \subseteq U$ is the set of elements that is not independent from $S$: $$\Span(S)=\{i \in U \mid \rank(S) = \rank(S \cup \{i\})\}.$$
    \item An independent set $S \in \indset$ is called a \emph{basis} if it spans the ground set, i.e., $\Span(S)=U$.
    \item The {\em restriction} of $\matroid$ to a set $S \subseteq U$ is a matroid 
    $\matroid|_{S}=(S,\indset |_{S}) = (S, \{I \in \indset \mid I \subseteq S\})$.
    \item The {\em contraction} of $\matroid$ by a set $S \subseteq U$ is a matroid $\matroid / S = (U \setminus S, \indset')$ where
    $$\indset'=\{I \in \indset |_{U \setminus S} \mid \rank(I) + \rank(S) = \rank(I \cup S)\}.$$
\end{itemize}
The polytope of a matroid $\matroid$ can be represented as $$\polytope_{\matroid} = \left\{\allocs \in \mathbb [0, 1]^{n} ~\middle|~ \sum_{i \in S} \alloci \leq \rank(S), \forall S \subseteq U\right\}.$$

%% file: matroid_prophet.tex
In this section, we reduce the matroid prophet inequality to matroid OCRS using only $O_{\varepsilon}(\log n)$ samples. Together with our $O_{\varepsilon}(\log^4 n)$-sample matroid OCRS, which we will present in \Cref{sec:orcs}, it proves the following theorem.

\begin{restatable}{theorem}{prophet} \label{thm:matroid_prophet}
With $O_\varepsilon(\log^4 n)$ samples, there is a $(\frac{1}{4}-\varepsilon)$-competitive prophet inequality for general matroids of size $n$ and any $\varepsilon>0$.
\end{restatable}

\subsection{Overview of Our Techniques} \label{sec:prophet-overview}

Before getting into our matroid prophet inequalities, let us first sketch the main idea and comparison with prior works. 

\paragraph*{Notations.} For simplicity, we now define a minimum set of notations that we will only use in this subsection. In the following, we assume that the items arrive in order $1, 2, \dots, n$ and $A_i$ is the set of items accepted by our prophet inequalities among the first $i$ items. 

For subsets $S,A \subseteq [n]$, we will use the random variable $$\OPT^{\vals{}}(S, A) \eqdef \max\{\val(I) \mid I \subseteq S, I \cup A \in \mathcal{I}\}$$ to denote the maximum-weight subset $I \subseteq S$, such that, $I \cup A$ is an independent set. For any independent set $I$, we define $$\tau_i^{\vals{}}(I) \eqdef \min\{v_j \mid j \in I, I - j + i \in \indset\}$$ to be the minimum value among the elements $j \in I$ that $i$ can exchange with. The existence of $j$ is guaranteed by the exchange property of matroids.

\paragraph*{Approaches in prior works.} Most matroid prophet inequalities follow a very general threshold-based framework:

\begin{itemize}
    \item Before the arrival of an item $i$, select a threshold $T_i$. If the value of $i$ exceeds the threshold, we call it an active item. 
    \item Then either accepts all active items whenever possible, or performs contention resolution over the active items.
\end{itemize}

Therefore, the key to designing matroid prophet inequalities is choosing the appropriate thresholds $\{T_i\}_{i \in [n]}$. Below, we list several popular choices. and explain why they break down in the sampled-based setting:

\begin{itemize}
    \item (\textit{Expectation-based.}) In the optimal matroid prophet inequalities by \citet{KW12}, $T_i$ is set to $\frac{1}{2} \E[\val]{\tau^{\vals{}}_i(\OPT^{\vals{}}([n] \setminus A_{i - 1}, A_{i - 1}))}$. After their arrivals, all active items are accepted.
    
    Such expectation-based thresholds $T_i$ are ubiquitous in the prophet inequalities and posted-price mechanisms literature (e.g. \citep{FGL14, GW19}). However, in sample-based setting, consider an item which has a large value $n^{10}$ with probability $1 / n$ and otherwise has a value of $0$. With only $O(\poly \log n)$ samples, one cannot hope to see it realized with the large value, and therefore cannot estimate the expectation. \\
    
    \item (\textit{Fresh-samples-based.}) In the reduction to OCRS by \citet{FSZ16}, they draw a fresh sample $\vals{}_{-i}' \sim \dist_{-i}$ for each arriving item $i$. The threshold $T_i$ is set to $\tau_i^{\vals{}'}(\OPT^{\vals{}'})$, where $\OPT^{\vals{}'} = \OPT^{\vals{}'}([n], \emptyset)$ is the maximum-weight independent set with respect to $\vals{}'$. Afterwards, the algorithm performs online contention resolution among those active items. 

    This approach is later used in prophet inequalities for $k$-uniform matroids~\citep{JMZ22, DW23} and bipartite matching~\citep{EFGT22}. By definition, each threshold in $\{T_i\}_{i \in [n]}$ requires one sample of the entire valuation vector $\vals$. In total, the reduction needs $\Omega(n)$ samples. \\

    \item (\textit{Single-sample-based.}) Naturally, one can try to instantiate the same fresh-samples-based thresholds with only a single sample $\vals{}'$ shared among all items. This idea was explored by \citet{Shaddin19} in matroid secretary, which is a slightly different setting.
    
    The problem is that OCRS requires the events $E_i$ that item $i$ is active to be independent for different $i \in [n]$. Using a single sample introduces positive correlations between these events. \citet{Shaddin19} shows that such positive correlated items still admit a constant-selectable CRS, but a constant-selectable OCRS for these items will imply the resolution of matroid secretary conjecture.
\end{itemize}

\paragraph{Our Approach. (Quantile-based.)} Historically, for single-item prophet inequalities, the seminal result by \citet{SC84} shows that the median of the maximum value distribution gives the optimal $2$-approximation. Such quantile-based thresholds are extended to the $k$-uniform matroids by \citet{chawla2020non}. Naturally, one might wonder: Is such a choice still applicable to general matroids? To the best of our knowledge, this question still remains unexplored (no matter whether you allow contention resolution at the end or not). %

In the sample-based setting, because previous thresholds break down, we are forced to understand quantile-based thresholds. Recall that $\OPT^{\vals{}} = \OPT^{\vals{}}([n], \emptyset)$ is the maximum-weight independent set with respect to random valuation $\vals$. In this work, we essentially set the threshold $T_i$ to be the (approximate) median of $\tau^{\vals{}}_i(\OPT^{\vals{}})$. (For simplicity, let us assume all $T_i$'s are the exact medians throughout this overview.) Then apply contention resolution. At first glance, this expression may look similar to the fresh-samples-based thresholds of \citet{FSZ16}. However, note that they are using \emph{the realization} of a fresh sample $\vals{}'$ for each threshold $T_i$. But here we estimate \emph{the medians} used to set all $T_i$'s using the same set of $O_{\varepsilon}(\log n)$ samples.

\paragraph{Main Difficulty.} To show that such thresholds, together with OCRS, provide a constant approximation, one needs to lower bound the total weight of active elements. That is, we wish to show that the contribution from nonactive items, $\sum_{i \in [n]} \vali \cdot \mathbbm{1}[\vali < T_i]$, is at most that from active items $\sum_{i \in [n]} \vali \cdot \mathbbm{1}[\vali \ge T_i]$. This turns out to be a non-trivial task. 

For example, let us examine the case of $k$-uniform matroids. In this case, $\tau^{\vals{}}_i(\OPT^{\vals{}})$ is always equal to the $k$-th largest value among all the items, denoted by $\val[k]$. Let $T'$ be the median of this value $\val[k]$. We first set a deterministic benchmark, $k \cdot T'$:
\begin{itemize}
    \item With probability $1/2$, we know that $\val[k] \ge T'$ and that the contribution of active elements $\sum_{i \in [n]} \vali \cdot \mathbbm{1}[\vali \ge T_i]$ is at least $k \cdot T'$ in this case. There is no contribution from nonactive items.
    \item With the other $1/2$ probability, we know that $\val[k] < T'$ and the contribution of nonactive items $\sum_{i \in [n]} \vali \cdot \mathbbm{1}[\vali < T_i]$ can be at most $k \cdot T'$ in this case. %
\end{itemize}

Thus, we know that $\sum_{i \in [n]} \vali \cdot \mathbbm{1}[\vali \ge T_i] \geq k \cdot T' / 2 \geq \sum_{i \in [n]} \vali \cdot \mathbbm{1}[\vali < T_i].$ We note that it is crucial that we are using a global argument: It first ``gathers'' the contribution of those elements $i$ with $\vali \ge T'$, comparing it against a benchmark $k \cdot T'$. Then it ``gathers'' all the contribution of those elements with $\vali < T'$ and upper bounds it with the benchmark. 

We note that such a global argument is necessary in a certain sense. Consider the following example:  We add a new element that is deterministically $T' - \varepsilon$ to the $k$-uniform matroid. It is always below the threshold $T'$, but is in the optimal solution with probability $1/2$. This results in a non-negligible contribution of $(T' - \varepsilon) / 2$. Since it is never active, we are forced to ``gather'' the contribution of (possibly many) other active elements to compare with it. 

Now, generalizing this argument to an arbitrary matroid, it is a priori not clear what the deterministic benchmark should be. We show that the weight of the optimal independent set w.r.t. weights $\{T_i\}_{i \in [n]}$, i.e. $\OPT^T$ is a good benchmark. To prove that $\sum_{i \in [n]} \vali \cdot \mathbbm{1}[\vali \ge T_i] \geq \OPT^T / 2 \geq \sum_{i \in [n]} \vali \cdot \mathbbm{1}[\vali < T_i]$, we need a weighted generalization of the strong basis exchange property (Lemma~\ref{lemma:monotone_mapping}), which was used in submodular maximization~\citep{BFG19} and has not been applied to prophet inequalities before. 

Finally, to avoid losing competitive ratio in this reduction, we slightly generalize this approach and set multiple thresholds using different quantiles.

\input{reduction}

\input{exchange}

%% file: reduction.tex
\subsection{Reduction from Matroid Prophet Inequality to OCRS with Samples}

\label{subsec:prophet}

Recall that we use $\matroid = (U, \indset)$ to denote a matroid (where $U=[n]$), and for each item $i \in U$, its value $\vali$ is drawn independently from $\disti$. For simplicity, we assume that all $\disti$'s are continuous distributions. For general distributions with point masses, a standard tie-breaking technique can be applied. For completeness, we include it in Appendix \ref{appendix:point_mass}. 

We first give some useful notations that we will use in this subsection. %
Given the realization of values $\vals$, we define $\OPT(\vals) \eqdef \max_{S \in \indset} \sum_{j \in S} \vali[j]$ as the value of the maximum independent set and $\OPT_{i}(\valsmi) \eqdef \max_{S \in \indset, i \notin S} \sum_{j \in S} \vali[j]$ as the optimal value when item $i$ is removed. We assume w.l.o.g. that every $\OPT_i(\valsmi)$ is a basis. We can always add dummy elements with zero value to ensure it. For simplicity of notation, we omit $\vals$ when there is no obfuscation and abuse $\opt, \opt_{i}$ to denote both optimal sets of items and their total values. 

Given $\valsmi$, define $\tau_i$ to be the smallest value among the items in $\OPT_{i}$ that $i$ can exchange with.
That is to say, $\tau_i(\valsmi) \eqdef \min\{\vali[j] \mid \OPT_{i} - j + i \in \indset \}$. 
Again, we omit $\valsmi$ when it is clear from the context.

The following observations regarding $\tau_i$ are useful in our later analysis.
\begin{observation}
\label{obs:is_threshold}
Fixing $\valsmi$, if $\vali > \tau_i(\valsmi)$, $i \in \opt(\vals)$; if $\vali < \tau_i(\valsmi)$, $i \notin \opt(\vals)$.\footnote{We ignore the equality case which happens with zero probability.}
\end{observation}
\begin{proof}
Let $j$ be the item whose value equals $\tau_i$. In the case when $\vali > \tau_i$, consider exchanging $i$ with $j$ in $\OPT_{i}$. This results in a basis of value $\OPT_{i} - \tau_i + \vali > \OPT_{i}$. Hence $\OPT_{i}$ is not optimal, and consequently $i \in \OPT$. 

For the other direction, we prove by contradiction. Suppose $i \in \OPT$, by Lemma \ref{lemma:strong_exchange}, there is an element $j \in \OPT_i$ such that $\OPT - i + j$ and $\OPT_i - j + i$ are both bases. Since $\OPT - i + j \in \indset$, by optimality of $\OPT$, $v_j \leq v_i$. Moreover, since $\OPT_i - j + i \in \indset$, by the definition of $\tau_i$, we have $\tau_i \leq v_j$. Hence $\tau_i > v_j$ implies $i \not\in \OPT$.
\end{proof}

\begin{observation}
\label{obs:tau_property}
Suppose $i \in \OPT$, $j \notin \OPT$ and $\OPT- i + j \in \indset$. Then $\vali \ge \tau_j$.
\end{observation}
\begin{proof}
Since $j \not\in \OPT$, we have $\OPT_j = \OPT$. Together with the fact that $\OPT - i + j$ is also an independent set, we have $\tau_j \leq \vali$ (by the definition of $\tau_j$). 
\end{proof}

Our algorithm consists of two stages. In the first stage, we estimate the thresholds $\{T^{(k)}_i\}$ from $O(\log n)$ samples. In the second stage, we use an OCRS to select a subset of active elements, which is sampled by an activation rule defined with thresholds. 
We present the first stage in the rest of this section. Let $m = \lfloor \log_{1 + \varepsilon}(1/\varepsilon) \rfloor$.

\paragraph{Learning Thresholds.} Given $N$ i.i.d.\ samples $\vals^s \sim \times_{i} \dist_i$, where $s \in [N]$.
For each item $i$, let $T^{(k)}_i$ be the $\lceil\varepsilon(1+\varepsilon)^kN\rceil$-th smallest value in $\{\tau_i(\valsmi^s) \mid s \in [N]\}$ for $0\le k< m$ and $T^{(m)}_i = \infty$.

The next lemma states that the thresholds $\{T^{(k)}_i\}_{i \in U}$ approximate the functions $p^i(v) \eqdef \Prx{i \in \opt \mid v_i=v}$ within $\varepsilon$ error. The proof is a direct application of Chernoff bound and we defer it to Appendix~\ref{appendix:threshold_concentration}.
Let $p_k = \varepsilon(1+\varepsilon)^k -\varepsilon^2$ for $0 \le k < m$. Note by Observation \ref{obs:is_threshold}, we have $p^i(v) = \Pr[v > \tau_i]$ for any fixed $v$, and $\{T^{(k)}_i\}$ are fixed values (instead of random variables) here. 

\begin{lemma} \label{lemma:threshold_concentration}
For any $\varepsilon \in (0,1)$, with $N=O(\log(\frac{2nm}{\varepsilon}) \cdot \varepsilon^{-4})$ samples, $p^i(T^{(k)}_i) = \Pr[T^{(k)}_i > \tau_i] \in [p_k,p_k+2\varepsilon^2]$ for all $i \in U$ and $0 \le k < m$ with probability at least $1 - \varepsilon$. We refer to such thresholds as \emph{good thresholds}.
\end{lemma}

Once we have estimated these thresholds, we view them as fixed value rather than random variables (i.e. conditioning on the value of these thresholds). This is important for the independence required by OCRSs.

\paragraph{Activation Rule.} With these thresholds, we specify a rule of determining whether an item $i$ is active. If $v_i \in [T^{(k)}_i, T^{(k+1)}_i)$ for some $ 0\le k<m$, let $i$ be active with probability $p_k$; else $v_i \in [0, T^{(0)}_i)$, let $i$ be active with probability $0$.
Let $\alloci = \Prx{i \text{ is active}}$, where the probability space only involves randomness of $v_i \sim \dist_i$.

For any $\vals \sim \times_i \dist_i$, we can then apply this activation rule to each of the items to generate a set of active items. It is important that in this stage, we do not reuse the samples from the first stage. The reason is that there are correlations between the thresholds we learnt for different items, but OCRS requires each item to be active independently. Thus we treat the thresholds now as fixed values and use the real realization of $\vals$ to generate $R(\allocs)$. Observe that, in this way, each element is active independently.

\begin{lemma}
	\label{lemma:ocrs_validity}
	If $\{T^{(k)}_i\}$ are good thresholds, then $\allocs \in \polytope_{\matroid}$.
\end{lemma}

\begin{proof}
	Let $\alloci^*=\Prx{i \in \OPT}$ for every $i \in U$. Since $\OPT \in \indset$, we have $\allocs^* \in \polytope_{\matroid}$. Consider the conditional probability of $i$ being active when $\vali = v$.
	\begin{itemize}
	\item If $v<T^{(0)}_i$, $\Prx{i \text{ is active} \mid v_i=v} = 0 \le p^i(v)$. 
	\item If $v \in [T^{(k)}_i, T^{(k+1)}_i)$ for $0\le k <m$, we have 
	\[
	\Prx{i \text{ is active} \mid \vali=v} = p_k \le p^i(T^{(k)}_i) \le p^i(v),
	\]
	where the first inequality follows from the definition of $T^{(k)}_i$ and the second inequality follows from the monotonicity of the function $p^i$.
	\end{itemize}
	Therefore, for each $i\in U$,
	\[
	\alloci = \int_v \Prx{i \text{ is active} \mid \vali = v} \dd \dist_i(v) \le \int_v p^i(v) \dd \dist_i(v) = \alloci^*,
	\]
	that concludes the proof of the lemma.
\end{proof}

Furthermore, the expected values of active items are large compared to the expected optimum. This is the most technical lemma and the proof is built on a monotone basis exchange lemma (Lemma~\ref{lemma:monotone_mapping}) for matroids from \cite{BFG19}.

\begin{lemma}
	\label{lemma:prophet_competitive_ratio}
	If $\{T^{(k)}_i\}$ are good thresholds, then $\E{\sum_{i \in R(\allocs)} v_i} \ge (1-O(\varepsilon)) \cdot \E{\opt}$.
\end{lemma}
\begin{proof}
 We defer its proof to \Cref{subsec:bases_exchange}. 
\end{proof}

Now, we have all the components to prove our main theorem.
\prophet*

\begin{proof}
We first use $O_{\varepsilon}(\log n)$ samples to learn the thresholds $\{T^{(k)}_i\}$. By Lemma~\ref{lemma:threshold_concentration} and Lemma~\ref{lemma:ocrs_validity}, the thresholds are good with at least $(1-\varepsilon)$ probability and the corresponding activation rule results in a valid $\allocs \in \polytope_{\matroid}$.
We then use $O_\varepsilon(\log^4 n)$ samples to learn the corresponding $(\frac{1}{4}-\varepsilon)$-selectable OCRS by Theorem~\ref{thm:matroid_ocrs}.
Finally, in the real run of the prophet inequality, upon the arrival of an item $i$, we use the activation rule constructed above to decide whether $i$ active and feed it to the OCRS.
	
Let $\ALG$ be the random set of elements selected by the algorithm. Let $\event$ be the event that $\{T^{(k)}_i\}$ are good thresholds. 
\begin{align*}
\E{\sum_{i \in \ALG} \vali} \ge{} & \Prx{\event} \cdot \operatorname{\mathbf E}\left[ \sum_{i\in \ALG} \vali \;\middle|\; \event \right] \\
\ge{} & (1-\varepsilon) \cdot \left(\frac{1}{4} - \varepsilon\right) \operatorname{\mathbf E}\left[ \sum_{i\in R(\allocs)}\vali \;\middle|\; \event \right] & \tag{Lemma~\ref{lemma:threshold_concentration} and Theorem~\ref{thm:matroid_ocrs}} \\
\ge{} & (1-\varepsilon) \cdot \left(\frac{1}{4} - \varepsilon \right) \cdot \left(1-O(\varepsilon) \right) \cdot \E{\opt} & \tag{Lemma~\ref{lemma:prophet_competitive_ratio}} \\
={} & \left(\frac{1}{4}-O(\varepsilon)\right) \cdot \E{\opt},
\end{align*}
that concludes the proof of the theorem.
\end{proof}

\begin{remark}[Against almighty adversary]\label{rem:almighty-adversary}
    Although here we state the prophet inequality problem with \emph{offline adversary}, i.e., assuming the items arrive in a fixed order $1,2,\dots,n$, our sample-based algorithm actually works against stronger \emph{almighty adversary}, who determines the arrival order adaptively with full knowledge of all realizations of randomness and the decisions made by the algorithm.
    The reason is that the thresholds are computed in a non-adaptive way (by using samples only), and the sample-based OCRS is a greedy OCRS, which works against almighty adversary (see \citep{FSZ16} for details).
\end{remark}

%% file: exchange.tex
\subsection{Proof of Lemma~\ref{lemma:prophet_competitive_ratio}}
\label{subsec:bases_exchange}

Observe that we select items in a conservative way, in the sense that $\Prx{i \text{ is active} \mid v_i =v} \le \Prx{i \in \opt \mid v_i=v}$ for all $i,v$. 
Nevertheless, we prove that the expected values of the active items are $\varepsilon$-close to the optimum.

We will need the following monotone basis exchange lemma for matroids. For completeness, we include a proof of it in \Cref{appendix:monotone-exchange}.

\begin{lemma}[Lemma 2.4, \citep{BFG19}]\label{lemma:monotone_mapping}
	Let $\matroid = (U, \indset)$ be a weighted matroid with weight function $w: U \rightarrow \mathbb{R}_{\ge 0}$. Suppose $A \in \mathcal{I}$ is the basis with maximum total weight. For every basis $B$, there exists a bijection $f: A \rightarrow B$ such that for all $x \in A$, $B - f(x) + x$ is a basis and $w(f(x)) \leq w(x)$. 
	Moreover, the bijection $f$ satisfies $f(x) = x$ for all $x \in A \cap B$. 
\end{lemma}

It is worthwhile to make a comparison to the following lemma in the literature. 

\begin{lemma}[Lemma 1, \citep{KW12}] \label{lemma:mapping_easy}
	Let $\matroid = (U, \indset)$ be a weighted matroid with weight function $w: U \rightarrow \mathbb{R}_{\ge 0}$. Suppose $A \in \indset$ is the basis with maximum total weight. For every basis $B$, there exists a bijection $g: B \rightarrow A$ such that for all $y \in B$, $A - g(y) + y$ is a basis and $w(g(y)) \geq w(y)$.
\end{lemma}

Observe that the only difference between the two lemmas is the direction of the mappings $f,g$. We note that for every bijection $g$ such that $A-g(y)+y$ is a basis, we must have $w(g(y)) \geq w(y)$ as otherwise $w(A - g(y) + y) > w(A)$. In contrast, Lemma~\ref{lemma:monotone_mapping} does not hold for every bijection $f$ such that $B-f(x)+x$ is a basis.

Now, we are ready to prove Lemma~\ref{lemma:prophet_competitive_ratio}.
The crucial part is to bound the contribution of those items with $\vali < T_i^{(0)}$.

\begin{lemma}\label{lemma:prophet_small_values}
	If $\{T_i^{(k)}\}$ are good thresholds, 
	\[
	\E{\sum_{i\in\OPT}\vali \cdot \I{\vali < T_i^{(0)}}} \le O(\varepsilon) \cdot \E{\opt}.
	\]
\end{lemma}

\begin{proof}
	Let $W=\argmax_{S\in\indset}\sum_{i\in S} T_i^{(0)}$ be the maximum weighted basis where item $i$ has value $T_i^{(0)}$.
	Fix arbitrary $\vals$ and $\OPT(\vals)$, by Lemma \ref{lemma:monotone_mapping}, there exists a bijection $f: W \rightarrow \OPT$ such that for any $j \in W$, $\OPT-f(j) + j \in \indset$ and $T_{f(j)}^{(0)} \le T_j^{(0)}$. Moreover, for all $j \in W \cap \OPT$, $f(j) = j$. 
	
	By Observation~\ref{obs:is_threshold}, for every $j \in W \cap \OPT$, we have $f(j)=j$ and $\vali[j] \ge \tau_j(\valsmi[j])$. For every $j \in W \setminus \OPT$, by Observation \ref{obs:tau_property}, we also have $\vali[f(j)] \ge \tau_j(\valsmi[j])$.
	Hence, $\vali[f(j)] \ge \tau_j$ for all $j \in U$. Consequently,
	\begin{equation}
	\Prx{\vali[f(j)] \ge T_j^{(0)}} \ge \Prx{\tau_j(\valsmi[j]) \ge T_j^{(0)}} \ge 1- \varepsilon - \varepsilon^2, \label{eqn:prob}
	\end{equation}
	where the last inequality holds when $\{T_i^{(k)}\}$ are good.
	Then, 
	\begin{equation}
	\E{\sum_{i\in\OPT}\vali} = \E{\sum_{j \in W} \vali[f(j)]} \ge \E{\sum_{j\in W} T_j^{(0)} \cdot \I{\vali[f(j)] \ge T_j^{(0)}}} \ge (1-\varepsilon-\varepsilon^2)\sum_{j\in W} T_j^{(0)}. \label{eqn:lowerbound}
	\end{equation}
	Recall that by the construction of $f$, we have that $T_{f(j)}^{(0)} \le T_j^{(0)}$, $\forall j \in W$. Equivalently, we have $T_{f^{\text{-}1}(i)}^{(0)} \ge T_i^{(0)}$, $\forall i \in \OPT$. Finally,
	\begin{align*}
	\E{\sum_{i\in\OPT}\vali \I{\vali < T_i^{(0)}}} \le{} & \E{\sum_{i\in\OPT}\vali \I{\vali < T_{f^{\text{-}1}(i)}^{(0)}}} \le \E{\sum_{i\in\OPT}T_{f^{\text{-}1}(i)}^{(0)} \I{\vali < T_{f^{\text{-}1}(i)}^{(0)}}}\\
	={} & \E{\sum_{j\in W} T_j^{(0)} \I{\vali[f(j)] < T_j^{(0)}}} \stackrel{\eqref{eqn:prob}}{\le} (\varepsilon+\varepsilon^2) \sum_{j\in W} T_j^{(0)} \\
	\stackrel{\eqref{eqn:lowerbound}}{\le}{} & \frac{\varepsilon+\varepsilon^2}{1-\varepsilon-\varepsilon^2} \E{\opt} \le O(\varepsilon) \cdot \E{\opt}. 
	\qedhere
	\end{align*}
\end{proof}

\begin{proof}[Proof of Lemma~\ref{lemma:prophet_competitive_ratio}]
For every $i \in U$ and $v \in [T_i^{(k)},T_{i}^{(k+1)})$, we have 
\begin{align*}
p^i(v) \le p^i(T_{i}^{(k+1)}) \le{} & \varepsilon(1+\varepsilon)^{k+1} + \varepsilon^2 & \tag{$\{T_{i}^{(k)}\}$ are good} \\
\le{} & (1+O(\varepsilon)) \cdot (\varepsilon(1+\varepsilon)^k - \varepsilon^2) \\
={} & (1+O(\varepsilon)) \Prx{i \in R(\allocs) \mid \vali=v}. & \tag{activation rule}
\end{align*}
Therefore,
\begin{align*}
\E{\sum_{i \in \OPT} \vali \I{\vali \ge T_i^{(0)}}} ={} & \sum_{i \in U} \int_{T_i^{(0)}}^{\infty} v p^i(v) \dd F_i(v) \\ 
\le{} & (1+O(\varepsilon)) \sum_{i \in [n]} \int_{T_i^{(0)}}^{\infty} v \Prx{i \in R(\allocs) \mid v_i = v} \dd F_i(v) \\ 
={} & (1 + O(\varepsilon))\E{\sum_{i \in R(\allocs)} \vali}.
\end{align*}

Combining this inequality with Lemma~\ref{lemma:prophet_small_values}, we conclude that
\begin{align*}
\E{\sum_{i\in R(\allocs)}\vali} \ge{} & \frac{1}{1+O(\varepsilon)} \E{\sum_{i\in\OPT}\vali\I{\vali\ge T_i^{(0)}}} \\
={} & \frac{1}{1+O(\varepsilon)} \left(\E{\opt} - \E{\sum_{i\in\OPT}\vali\I{\vali < T_i^{(0)}}} \right) \ge \frac{1}{1+O(\varepsilon)} \E{\opt}. \qedhere
\end{align*}
\end{proof}

%% file: matroid_ocrs.tex
Motivated by our sample-based reduction, our goal now is to design a constant selectable matroid OCRS using polylogarithmic samples. A previous work~\citep{HPT22} studied lower bounds for this problem.  Upon careful inspection, their counterexample shows that any constant selectable matroid OCRS requires $\Omega(\log n)$ samples. (Details are given in Appendix~\ref{appendix:ocrs-sample-lower-bound}.) In this work, we design a matroid OCRS with $O_{\varepsilon}(\log^4 n)$ samples:

\begin{theorem}\label{thm:matroid_ocrs}
With $O_{\varepsilon}(\log^4 n)$ samples, there is a $(\frac{1}{4} - \varepsilon)$-selectable OCRS for general matroids of size $n$ and any $\varepsilon > 0$.
\end{theorem}

Our OCRS is based on the matroid OCRS proposed by~\citet{FSZ16}. Recall that in online contention resolution, each element $e \in N$ is active independently with probability $\alloc_e$. Let $R(\allocs)$ be the random set of active elements. For sample-based schemes, this vector $\allocs$ is hidden, instead the algorithm gets $s$ samples of $R(\allocs)$. To implement the algorithm in~\citep{FSZ16}, it suffices to implement an oracle that answers queries of form $\Pr[e \in \Span(R(\allocs) \cup S)]$ for element $e \in N$ and subset $S \subseteq N$ with error up to $\varepsilon$.

For a single fixed query, evaluating the empirical probability over $\frac{\log(n)}{\varepsilon^2}$ samples guarantees the success with high probability. At first sight, it is tempting to argue that these samples suffice by a union bound over all $\poly(n)$ queries made by their algorithm. However, such argument is only true for non-adaptive queries. For adaptive queries, the error in the answers to previous queries can reveal information about the samples that leads to a problematic query. Resolving this issue is a central problem in adaptive data analysis. As the OCRS by~\citet{FSZ16} dynamically updates the set $S$ based on previous queries, it makes $\poly(n)$ highly adaptive queries. Studies in adaptive data analysis show that, in general, answering these many adaptive queries up to $\varepsilon = 0.1$ error already requires $\poly(n)$ samples~\citep{HU14, SU15}. In our work, we protect against such adaptivity with new ideas that are specific to OCRSs.

\subsection{Recap: Chain Decomposition}

In this section, we briefly sketch the idea of the OCRS by~\citet{FSZ16}. We are given a matroid $\matroid = (U, \mathcal{I})$, where $U = [n]$ is the universe, and $\mathcal{I}$ is the family of independent sets. Suppose $\allocs \in b \cdot \polytope_\matroid$ is a vector within the matroid polytope shrank by a factor of $b$, where $b$ is a constant factor that we will pick later. Each element $e \in U$ is active independently with probability $\alloc_e$. Let $R(\allocs)$ be the random set of active elements (w.r.t. the vector $\allocs$).

\paragraph{Protected elements.} Consider the greedy algorithm that always takes an active element whenever possible. Its selectability for an element $e \in U$ can be lower-bounded: $$\Prx{e\text{ is accepted by greedy} \given e \text{ is active}} \geq 1 - \Pr[e \in \Span(R(\allocs) \setminus \{e\})].$$ The idea of chain decomposition is that we are going to protect those elements $e \in U$ with a selectability less than $1 - c$. We prioritize those elements by adding them to a set of protected elements $S$. Now the modified greedy algorithm is going to pretend that all elements in $S$ are already taken at the beginning, then take any active element whenever possible. As the set $S$ is now nonempty, the selectability for an element $e \in U \setminus S$ becomes
$$\Prx{e\text{ is accepted by modified greedy} \given e \text{ is active}} \geq 1 - \Pr[e \in \Span(R(\allocs) \cup S \setminus \{e\})].$$
There may be new problematic elements with this probability less than $1 - c$. We then have to enlarge the set $S$ by adding those elements. In this way, the set $S$ is \emph{dynamically} changing based on the result of the previous queries of such a probability. Then new queries depend on the changing $S$, which makes them highly \emph{adaptive} in nature. This process is summarized in Algorithm~\ref{alg:protection}. It takes a submatroid $(N, \mathcal{I} \vert_N)$ for $N \subseteq U$ and returns the protected set $S$ for that submatroid. We note that although the choice of the element $e$ at Line~\ref{line:choose_element} is arbitrary, the returned set $S$ is always unique. (See Appendix~\ref{appendix:ocrs-protection-uniqueness}.)

\begin{algorithm}
\SetKwFunction{select}{$\mathsf{Select}$}
\SetKwProg{myfunc}{Function}{}{}
\myfunc{\select{$N$, $c$}}{
$S \gets \emptyset$

\While{$\exists e \in N\setminus S, \ \Prx{e \in \Span(((R(\allocs) \cap N) \cup S) \setminus \{e\})} > c$ \label{line:choose_element}} {add such element $e$ to $S$.} 

\Return{$S$}
}

\caption{Select protected elements for $N$}
\label{alg:protection}
\end{algorithm}

\paragraph{Chain decomposition} After selecting the protected set $S$, note that the elements outside $S$ can be handled by the modified greedy algorithm and achieve $(1 - c)$-selectability. Furthermore, the following lemma guarantees that set $S$ returned by $\select(N, c)$ is always of a smaller rank than $N$ when $b \leq c$. 

\begin{lemma}[Section 2.1.1, \citep{FSZ16}] \label{lemma:OCRS:MatroidTerminate}
For any submatroid $(N, \mathcal{I}\vert_N)$, vector $\allocs \in b \cdot \polytope_\matroid$, and $c > 0$, we always have 
$$\rank(S) < \frac{b}{c} \cdot \rank(N)\text{ where }S = \select(N, c).$$
\end{lemma}

\begin{proof}
 We include a proof in~\Cref{appendix:ocrs-termination} for completeness. 
\end{proof}

Therefore, as long as we pick $c \geq b$, the subproblem in $S$ is always strictly smaller. We can recursively apply the idea to the submatroid $(S, \mathcal{I}\vert_S)$. This leads to the following algorithm:

\begin{algorithm}
\SetKwFunction{decompose}{$\mathsf{Decompose}$}
\SetKwProg{myproc}{Procedure}{}{}
\myproc{\decompose{$U$, $c$}}{
$N_0 \gets U$ 

$\ell \gets 0$

\While {$N_\ell \neq \emptyset$} {
$N_{\ell + 1} \gets \ $\select{$N_l$, $c$}

$\ell \gets \ell + 1$
}
}
\caption{Chain decomposition in \citet{FSZ16}}
\label{alg:chain-decomposition}
\end{algorithm}

By Lemma~\ref{lemma:OCRS:MatroidTerminate}, when $b < c - \varepsilon$, it is guaranteed to terminate in $\ell = O_{\varepsilon}(\log n)$ steps and produce the following chain decomposition
$$\emptyset = N_\ell \subsetneq N_{\ell-1} \subsetneq \cdots \subsetneq N_1 \subsetneq N_0 = U ~~~\text{where}~~~ N_{i+1} = \select(N_i,c) \text{ for } 0 \le i < \ell.$$
The final OCRS in~\citep{FSZ16} is simply running Algorithm~\ref{alg:modified-greedy}, a modified greedy algorithm, for each layer $N_i \setminus N_{i+1}$ separately.

\begin{algorithm}
$A \gets \emptyset$ 
\tcp*{Here $A$ is the set of accepted elements.}

\For{each arriving element $e \in N_i \setminus N_{i+1}$} {
\If {$A \cup N_{i + 1} \in \mathcal{I}$} {
$A \gets A \cup \{e\}$
} 
}
\Return{$A$}

\caption{Modified greedy algorithm for $N_i \setminus N_{i+1}$}
\label{alg:modified-greedy}
\end{algorithm}

In other words, \Cref{alg:modified-greedy} is equivalent to a greedy algorithm running on the matroid \(\matroid |_{N_i} / N_{i + 1}\), which is $\matroid$ restricted to $N_i$ and then contracted by $N_{i+1}$.
It is easy to see that running \Cref{alg:modified-greedy} for every $0\le i<\ell$ together always produces an independent set of $\matroid$.
And the selectability of any element $e \in N_i \setminus N_{i+1}$ is
\begin{equation*}
	\Prx{e\text{ is accepted} \given e \text{ is active}} \geq 1 - \Prx{e \in \Span(((R(\allocs) \cap N_i) \cup N_{i+1})) \setminus \{e\})} \ge 1-c.
\end{equation*}

\subsection{Selecting Protected Elements with Samples}

Now, we move on to introduce our sampled-based OCRS. As a natural first step, let us first try to implement Algorithm~\ref{alg:protection} with samples: 

\begin{algorithm}[H]
\SetKwFunction{sselect}{$\widehat{\mathsf{Select}}$}
\SetKwProg{myfunc}{Function}{}{}
\myfunc{\sselect{$N$, $c$}}{
Let $R_1, R_2, \dots, R_s$ be $s$ \emph{fresh} realized samples from $R(\allocs)$

For any event $E(R)$, define $\widehat{\mathbf{Pr}}[E] \coloneqq \frac{1}{s}\sum_{i=1}^s \mathbbm{1}[E(R_i)]$

$\widehat{S} \gets \emptyset$

\While{$\exists e \in N\setminus S,  \widehat{\mathbf{Pr}}[e \in \Span(((R \cap N) \cup \widehat{S}) \setminus \{e\})] > c$} {add such element $e$ to $\widehat{S}$}

\Return{$\widehat{S}$}
}

\caption{Select protected elements for $N$, with samples}
\label{alg:sample-select}
\end{algorithm}

Suppose that the algorithm returns a final set $\widehat{S}$. In order for the chain decomposition approach to work, we need the following property: For some constant $c'$ strictly less than $1$, 
\begin{equation}
\Prx[R(\allocs)]{e \in \Span(((R(\allocs) \cap N)\cup \widehat{S}) \setminus \{e\})} \le c', \ \ \  \forall e \in N \setminus \widehat{S}.\label{equ:property}
\end{equation}

It is tempting to think that if we take $s = \frac{\log n}{\delta^2}$ samples, by concentration, we will be able to use $\widehat{\mathbf{Pr}}[E]$ to answer polynomially many $\mathbf{Pr}[E]$ queries up to $\delta$ error with high probability, which ensures this property with $c' = c - \delta$. This is only true for non-adaptive queries. Here, in Algorithm~\ref{alg:sample-select}, it is crucial that the algorithm maintains a dynamic set $\widehat{S}$ and adaptively discovers new elements that need to be protected. The issue with adaptive queries is that the error in the previous answers reveals information about the samples. 

Looking closer, Algorithm~\ref{alg:sample-select} compares the answer to those $\mathbf{Pr}[E]$ queries with a threshold $c$. Consider a single $\mathbf{Pr}[E]$ query for estimating $\Prx{e \in \Span(((R(\allocs) \cap N)\cup \widehat{S}) \setminus \{e\})}$. If the true answer is less than $c - \delta$ (or greater than $c + \delta$), with high probability, the estimate returned by $\widehat{\mathbf{Pr}}$ will also be below (or above) the threshold $c$. Adaptivity appears when the true answer lies within $[c- \delta, c + \delta]$, whether the element $e$ is added to $\widehat{S}$ or not depends now on the error on the samples. It affects the set $\widehat{S}$, which determines the later queries.

\begin{remark}[Adaptive Data Analysis]
$\widehat{\mathbf{Pr}}[E]$ asks for the fraction of samples that satisfies an event $E$. These are called statistical queries. Answering statistical queries against adaptivity is a central problem in the field of adaptive data analysis. Lower bounds from there show, in general, with $s$ samples, one can answer at most $s^2$ many such adaptive queries~\citep{HU14, SU15}. There are a few exceptions, e.g. when the support of the sample is polynomially bounded \citep{bassily2016algorithmic}, or when all queries are threshold queries, and they are sparse \citep{dwork2009complexity, hardt2010multiplicative} (in the that sense that, only a bounded number of queries have above-threshold answers, which can be solved using the sparse vector technique). However, in our case, the support of $R(\allocs)$ can be as large as $2^n$. Although the queries are threshold queries, they are not sparse. As these existing tools seem insufficient, novel ideas are needed to resolve the adaptivity issue in Algorithm~\ref{alg:sample-select}. 
\end{remark}

To handle the adaptivity issue, let us observe a surprising property of $\sselect(N, c)$: despite the presence of adaptivity, $\sselect(N, c)$ is still sandwiched between two sets that satisfy Property~\eqref{equ:property}.

\begin{lemma}\label{lemma:bound-sample-based-protection}
    Fix any set $N \subseteq U$ and constants $0 < c_1 < c_2 < 1$. We define $\widehat{c} = \frac{c_1 + c_2}{2}$ and $\delta = \frac{c_2 - c_1}{2}$. Using $O\left(\frac{\log n\log\varepsilon^{-1}}{\delta^2}\right)$ samples, $$\select(N,c_2) \subseteq \sselect(N,\widehat{c}) \subseteq \select(N,c_1)$$ holds with probability at least $1-\frac{\poly(\varepsilon)}{\poly(n)}$.
\end{lemma}

\begin{proof}
    We prove these two cases separately:

    \begin{itemize}
        \item $\sselect(N,\widehat{c}) \subseteq \select(N,c_1)  $.

        Let $S_1 = \select(N,c_1)$. Consider the process of $\sselect(N,\widehat{c})$ which dynamically maintains the set $\widehat{S}$. We claim that $\widehat{S} \subseteq S_1$ throughout the process. We prove this via induction. Initially, we have $\widehat{S} = \emptyset \subseteq S_1$. Then at any step, if we are going to add $e$ to $\widehat{S}$, we must have $$\widehat{\Pr}[e \in \Span(((R(\allocs) \cap N) \cup \widehat{S} ) \setminus \{e\})] > \widehat{c}.$$
        By induction hypothesis, we know that $\widehat{S} \subseteq S_1$, which implies $(R(\allocs) \cap N) \cup \widehat{S} \subseteq (R(\allocs) \cap N) \cup S_1$. Hence, by monotonicity, $$\widehat{\Pr}[e \in \Span(((R(\allocs) \cap N) \cup S_1 ) \setminus \{e\})]] > \widehat{c}.$$
        Note $S_1 = \mathrm{Select}(N,c_1)$ is a fixed set that do not depend on samples. As we take $O\left(\frac{\log n \log\varepsilon^{-1}}{\delta^2}\right)$ samples, by Chernoff bound, with probability at least $1-\frac{\poly(\varepsilon)}{\poly(n)}$, we have 
        \begin{equation} \label{equ:estimation-S1}
        \left|\widehat{\Pr}[e \in \Span(((R(\allocs) \cap N) \cup S_1 ) \setminus \{e\})]] - \Pr[e \in \Span(((R(\allocs) \cap N) \cup S_1 ) \setminus \{e\})]\right| \leq \delta
        \end{equation}
        A union bound over all $n$ possibilities of $e$ shows that with probability $1-\frac{\poly(\varepsilon)}{\poly(n)}$, Equation~\eqref{equ:estimation-S1} holds throughout the process. Hence, for the newly added element $e$, we have $$\Pr[e \in \Span(((R(\allocs) \cap N) \cup S_1 ) \setminus \{e\})] > \widehat{c} - \delta \geq c_1.$$
        Since all elements in $U \setminus S_1$ are $(1 - c_1)$-selectable, we know that $e$ must be in $S_1$.

        \item $ \select(N,c_2) \subseteq \sselect(N,\widehat{c}) $.

        First of all, let $\widehat{S} = \sselect(N,\widehat{c})$. When \Cref{alg:sample-select} terminates, we know that for all $e \in U \setminus \widehat{S}$, $$\widehat{\mathbf{Pr}}[e \in \Span(((R \cap N) \cup \widehat{S}) \setminus \{e\})] \leq c.$$
        Consider the process of $\select(N, c_2)$, which dynamically maintains a set $S_2$. We claim that with high probability over the randomness of $\widehat{S}$, it holds that $S_2 \subseteq \widehat{S}$ through out this process. Again, we prove this by induction. Initially, we have $S_2 = \emptyset \subseteq \widehat{S}$. At any step, if we are going to add $e$ to $S_2$, we must have 
        $$\Pr[e \in \Span(((R(\allocs) \cap N) \cup S_2) \setminus \{e\})] > c_2.$$

        Note that the process of generating $S_2$ (i.e., the procedure $\select(N,c_2)$) does not depend on our samples. Similar to the previous case, by Chernoff bound, we get that with probability at least $1-\frac{\poly(\varepsilon)}{\poly(n)}$, 
        $$\widehat{\Pr}[e \in \Span(((R(\allocs) \cap N) \cup S_2) \setminus \{e\})] > c_2 - \delta > c_2 - \delta > \widehat{c}.$$

        By the induction hypothesis, we know that with $1-\frac{\poly(\varepsilon)}{\poly(n)}$ probability, $S_2 \subseteq \widehat{S}$, which implies $(R(\allocs) \cap N) \cup S_2 \subseteq (R(\allocs) \cap N) \cup \widehat{S}$. Hence, with the same high probability, 
        $$\widehat{\Pr}[e \in \Span(((R(\allocs) \cap N) \cup \widehat{S}) \setminus \{e\})] > \widehat{c}.$$

        If this happens, $e$ must be in the set $\widehat{S}$. Finally, we finish the proof by observing that we use union bound over at most $\poly(n)$ induction steps. \qedhere

    \end{itemize}
\end{proof}

According to Lemma~\ref{lemma:bound-sample-based-protection}, if we let $S_{c - \delta} = \select(N,c - \delta)$ and $S_{c + \delta}= \select(N,c + \delta)$, the set $\widehat{S}_c=\widehat{S}(N,c)$ must satisfy $S_{c + \delta} \subseteq \widehat{S}_c \subseteq S_{c - \delta}$ with high probability. Note that $S_{c - \delta}$ and $S_{c + \delta}$ are both valid choices for the set $S$ that satisfy \Cref{equ:property} with different constants. Since $\widehat{S}_c$ is sandwiched between them, it is very tempting to think that the set $\widehat{S}_c$ must also satisfy \Cref{equ:property}. If so, we would have proved that the original OCRS of \citet{FSZ16} (\Cref{alg:sample-chain-decomposition}) works with samples. However, this is not the case. 

To see this, consider the quantity $r_e \coloneqq \Pr[e \in \Span(((R(\allocs) \cap N)\cup \widehat{S}_c) \setminus \{e\})]$ for every element $e \in N \setminus \widehat{S}_c$: 
\begin{itemize}
    \item It cannot be the case that $e \in S_{c + \delta}$. If this is the case, we are guaranteed that $e \in \widehat{S}_c$ as well, which contradicts $e \in N \setminus \widehat{S}_c$.
    \item If $e \in N \setminus S_{c - \delta}$, we are guaranteed by \Cref{alg:protection} that $$\Pr[e \in \Span(((R(\allocs) \cap N)\cup S_{c - \delta}) \setminus \{e\})] \leq c - \delta.$$ Since $S_c \subseteq S_{c - \delta}$, by monotonicity, the quantity $r_e$ that we care about is also at most $c - \delta$.
    \item If $e \in S_{c - \delta} \setminus S_{c + \delta}$, on the one hand, $e$ is in $S_{c - \delta}$, so we do not have any non-trivial upper bound for $$\Pr[e \in \Span(((R(\allocs) \cap N)\cup S_{c - \delta}) \setminus \{e\})].$$
    On the other hand, $e \in N \setminus S_{c + \delta}$, although we know that $$\Pr[e \in \Span(((R(\allocs) \cap N)\cup S_{c + \delta}) \setminus \{e\})] \leq c + \delta.$$ But $\widehat{S}_c$ is (potentially) a larger set than $S_{c + \delta}$. It could be the case that $r_e$ is much larger than $c + \delta$. So we cannot hope to get a non-trivial upper bound for $r_e$.
\end{itemize}

In conclusion, the only problematic case is when $e \in S_{c - \delta} \setminus S_{c + \delta}$, where we do not have any guarantee of their selectability.

\subsection{Against Adaptivity using Randomization}

To illustrate our idea, we will first present a simple toy example to warm up. Then we will slightly generalize it to our actual construction.

\paragraph{Toy Example.} Our main idea for resolving such an adaptivity issue is simple to state: by choosing the protected set $\widehat{S}$ to be either $\widehat{S}_c$ or $\widehat{S}_{c+2\delta}$ uniformly at random, i.e.
$$\widehat{S} = \begin{cases}
    \sselect(N,c)& \text{w.p.}~\frac{1}{2},\\
    \sselect(N,c+2\delta)& \text{w.p.}~\frac{1}{2}.\\
\end{cases}$$
By Lemma~\ref{lemma:bound-sample-based-protection} and union bound, we know with high probability,
$$S_{c+3\delta} \subseteq \widehat{S}_{c+2\delta} \subseteq S_{c+\delta} \subseteq \widehat{S}_c \subseteq S_{c-\delta}.$$
We know if one chooses to protect $\widehat{S}_c$ (i.e., $\widehat{S} =  \sselect(N,c)$), we have no guarantee for the elements in $S_{c-\delta} \setminus S_{c+\delta}$. Similarly, if one chooses to protect $\widehat{S}_{c + 2\delta}$ (i.e., $\widehat{S} =  \sselect(N,c + 2\delta)$), elements in $S_{c+\delta} \setminus S_{c+3\delta}$ are problematic. So neither of these two works. But the crucial observation is that $S_{c-\delta} \setminus S_{c+\delta}$ and $S_{c+\delta} \setminus S_{c+3\delta}$ are two disjoint sets. If we just randomize between these two, we can imagine having guarantees for all elements with only a factor-of-$2$ loss in selectability.  

More specifically, for any element $e \in N$, 
\begin{itemize}
    \item with probability $1/2$, it is not in the problematic set of elements, which means with high probability, either $e \in \widehat{S}$ or $\Prx[R(\allocs)]{e \in \Span(((R(\allocs) \cap N)\cup \widehat{S}) \setminus \{e\})}  \leq c + 3\delta$;
    \item with the other $1/2$ probability, it is in the problematic set and we do not have any guarantee. 
\end{itemize}

Note that this only achieves \Cref{equ:property} with probability $1/2$. There is still a probability of $1/2$ that a bad event occurs: $e \in N \setminus \widehat{S}$ but $\Prx[R(\allocs)]{e \in \Span(((R(\allocs) \cap N)\cup \widehat{S}) \setminus \{e\})}$ is too large. 

This becomes a problem if we plug this procedure into the chain decomposition procedure (\Cref{alg:chain-decomposition}). From \Cref{lemma:OCRS:MatroidTerminate}, we know that the chain decomposition procedure has at most $\ell = O_{\varepsilon}(\log n)$ rounds. So it calls the procedure $\sselect$ $\ell$ times. Let us imagine the worst case: For every round in which $e$ is not in the problematic set, let us say that $e$ is in $\widehat{S}$ and the decomposition procedure continues with $e \in N$. Once we meet a round where $e$ is in the problematic set (with probability $1/2$), we lose all guarantees for the element $e$. But it is highly likely (with probability $1 - (1/2)^{O_{\varepsilon}(\log n)}$) that we will see one such problematic round.

\paragraph{Our Construction.} In order to fix this issue and achieve constant selectability, our final scheme needs a sequence $c_1 < c_2 < \dots < c_{k + 1}$ of $k + 1$ different thresholds that are equally spaced in $[c_1, c_{k + 1}]$. The scheme picks $j \in [k]$ uniformly at random and then uses $\widehat{c} = \frac{c_j + c_{j + 1}}{2}$. Note that for each round $i$, we are using $s$ fresh samples for $\sselect(\widehat{N}_i, \widehat{c})$. Finally, we are going to take $k = O_{\varepsilon}(\log n)$, so that the probability that the bad event happens, $1 - (1 - 1/k)^{O_{\varepsilon}(\log n)}$, is small enough. 

\begin{algorithm}
\SetKwFunction{sdecompose}{$\widehat{\mathsf{Decompose}}$}
\SetKwProg{myproc}{Procedure}{}{}
\myproc{\sdecompose{$U$, $c_1$, $c_2$, $\dots$, $c_{k + 1}$}}{
$\widehat{N}_0 \gets U$ 

$\ell \gets 0$

\While {$\widehat{N_\ell} \neq \emptyset$} { 
Sample $j_\ell \sim \mathrm{Uniform}([k])$ and let $\widehat{c} \gets \frac{c_{j_\ell} + c_{j_\ell + 1}}{2}$

$\widehat{N}_{\ell + 1} \gets \ $\sselect{$\widehat{N_\ell}$, $\widehat{c}$} %

$\ell \gets \ell + 1$
}
}
\caption{Chain decomposition, with samples}
\label{alg:sample-chain-decomposition}
\end{algorithm}

\subsection{Proof of \Cref{thm:matroid_ocrs}}

In this section, we complete the proof of Theorem~\ref{thm:matroid_ocrs} using~\Cref{alg:sample-chain-decomposition}. In our proof, we set \(b = \frac12\), which means that we shrink the ex-ante active probability $\allocs$ by half to ensure $\allocs \in \frac12 \cdot \polytope_\matroid$. Let 
\[\varnothing = \widehat{N}_\ell \subsetneq \widehat{N}_{\ell - 1} \subsetneq \cdots \subsetneq \widehat{N}_0 = U=[n]\]
be the chain decomposition generated by Algorithm~\ref{alg:sample-chain-decomposition}. 
Based on such decomposition, we aim to demonstrate that running the modified greedy algorithm, \Cref{alg:modified-greedy}, for each layer $\widehat{N}_i \setminus \widehat{N}_{i+1}$ separately achieves \((\frac{1}{4} - \varepsilon)\)-selectability using \( O_{\varepsilon}(\log^4 n) \) samples.

We first delineate the parameters for \Cref{alg:sample-chain-decomposition}, setting $k$ as $\frac{\log n}{\log(1 + \varepsilon) \log_{\frac{1}{4}}(1 - \varepsilon)} = \Theta_{\varepsilon}(\log n)$ and defining $\{c_j\}_{j = 1}^{k + 1}$ as an arithmetic sequence with initial term $c_1=\frac{1}{2} + \frac{\varepsilon}{2}$ and terminal term $c_{k+1}=\frac{1}{2} + \varepsilon$. In each layer $i$, \(s = \Theta\left(\frac{k^2 \log n \log\varepsilon^{-1}}{\varepsilon^2}\right) = \Theta_{\varepsilon}\InParentheses{\log^3 n}\) fresh samples are needed for $\sselect{$\widehat{N}_i, \widehat{c}$}$. %

\begin{lemma}
\label{lem:sample-chain-decomposition-property}
Let \(\ell\) denote the number of layers in the chain decomposition generated by \Cref{alg:sample-chain-decomposition}, with parameters set as above. The following holds with probability at least \(1 - \frac{\poly(\varepsilon)}{\poly(n)}\):
\begin{itemize}
    \item \(\ell = O_{\varepsilon}(\log n).\)
    \item \(\InParentheses{1 - \frac{1}{k}}^\ell \geq 1 - \varepsilon\).
\end{itemize}
\end{lemma}
\begin{proof}

Consider any protected set \(\widehat{N}_i\) and an index \(j_i\) uniformly selected from \([k]\) in layer \(i\). Let \(\delta = \frac{\varepsilon}{4k}\) represent the gap between \(\frac{c_{j_i} + c_{j_i + 1}}{2}\) and both \(c_{j_i}\) and \(c_{j_i+1}\). Given the parameters previously defined, we take \(s = \Theta\left(\frac{k^2 \log n \log \varepsilon^{-1}}{\varepsilon^2}\right) = \Theta\left(\frac{\log n \log \varepsilon^{-1}}{\delta^2}\right)\) new samples during the execution of $\sselect{$\widehat{N}_i, \widehat{c}$}$.  By applying \Cref{lemma:bound-sample-based-protection}, it is established that with probability at least \(1 - \frac{\poly(\varepsilon)}{\poly(n)}\):
 \[\select{$\widehat{N}_i, c_{j_i + 1}$} \subseteq \widehat{N}_{i+1} = \sselect{$\widehat{N}_i, \frac{c_{j_i} + c_{j_i + 1}}{2}$} \subseteq \select{$\widehat{N}_i, c_{j_i}$} .\]

It always holds that \(\select{$\widehat{N}_i, c_{j_i}$} \subseteq \select{$\widehat{N}_i, c_1$}\). Therefore, 
\begin{align*}
\begin{split}
    \rank\InParentheses{\widehat{N}_{i+1}} &\leq \rank\InParentheses{\select{$\widehat{N}_i, c_1$}}\\
        & < \frac{1}{1 + \varepsilon}\rank\InParentheses{\widehat{N}_i}.
\end{split} 
\end{align*}where the last inequality follows from \Cref{lemma:OCRS:MatroidTerminate} with \(b = \frac{1}{2}\) and \(c = \frac{1}{2} + \frac{1}{2}\varepsilon\). Since \Cref{alg:sample-chain-decomposition} terminates when \(\rank(\widehat{N}_\ell) < 1\), it follows that 
\begin{align*}
    \ell \leq \log_{1 + \varepsilon} \rank\InParentheses{\widehat{N}_0} \leq \frac{\log n}{\log \InParentheses{1 + \varepsilon}} = O_{\varepsilon}(\log n).
\end{align*}

Now, to demonstrate the second property, recall that \(k\) is defined as \(\frac{\log n}{\log(1 + \varepsilon) \log_{\sfrac{1}{4}}(1 - \varepsilon)}\). It thus follows that
\begin{equation}\label{eq:k-large-enough}
    \frac{\ell}{k} \leq \log_{\sfrac14}{(1-\varepsilon)}.
\end{equation}

Therefore,
\begin{align*}
    \left(1 - \frac{1}{k}\right)^\ell & = \left(\left(1 - \frac{1}{k}\right)^k\right)^{\frac{\ell}{k}}
     \geq \left(\frac{1}{4}\right)^{\frac{\ell}{k}}
    \geq 1 - \varepsilon,
\end{align*}
where the first inequality leverages the property that \(\left(1 - \frac{1}{k}\right)^k\) is an increasing function for \(k \geq 2\), and \(k\) is trivially at least \(2\). The final inequality is justified by \Cref{eq:k-large-enough}.
\end{proof}

Our algorithm employs \(O_{\varepsilon}\InParentheses{\log^3 n}\) samples per layer across \(\ell\) layers in total. Combining with \Cref{lem:sample-chain-decomposition-property}, it is evident that the total number of samples required by our algorithm does not exceed \(O_{\varepsilon}\InParentheses{\log^4 n}\).\footnote{Should the conditions specified in \Cref{lem:sample-chain-decomposition-property} not be met, \Cref{alg:sample-chain-decomposition} terminates immediately and announces failure. This approach has a negligible impact on the algorithm's selectability (in the order of \(\frac{\poly(\varepsilon)}{\poly(n)}\)) and guarantees that the algorithm deterministically requires \(O_{\varepsilon}(\log^4 n)\) samples.}

The remaining task is to demonstrate the selectability of \Cref{alg:sample-chain-decomposition}. Consider any element \(e \in U\); our objective is to establish that \[\Pr[e \text{ is accepted } \mid e \text{ is active}] \geq \frac{1}{2} - O(\varepsilon).\] Given that \(\allocs\) is contained within \(\frac{1}{2} \cdot \polytope_{\matroid}\), it is deduced that our algorithm achieves a selectability of \(\frac{1}{4} - O(\varepsilon)\). For this purpose, we will introduce a stochastic process to simulate the execution of \Cref{alg:sample-chain-decomposition} with respect to a fixed element $e$.

\paragraph{Stochastic Process} 
Fix an element $e \in U$.
Initiating with \(i = 0\) and \(\widehat{N_0} = U=[n]\), our stochastic process sequentially reveals the randomness, i.e., \(j_0, j_1,\dots,j_\ell\), across each layer. Assuming the process is at layer \(i\), with the protected set \(\widehat{N}_i\) containing element \(e\), we proceed to define $N_{i + 1, t} = \select{$\widehat{N}_i, c_t$} $ for each \(t \in [k + 1]\). Specifically, we set \(N_{i + 1, 0} = \widehat{N}_i\) and \(N_{i + 1, k + 2} = \varnothing\). It is evident that the sequence \(\{N_{i, t}\}_{t=1}^{k+2}\) satisfies the following property:
\[\varnothing = N_{i, k+2} \subseteq N_{i, k+1} \subseteq \cdots \subseteq N_{i, 1} \subseteq N_{i, 0} = \widehat{N}_i.\]

Given that \(e\in \widehat{N}_i\), there must exist a unique \(t^{*} \in \{0, 1,\dots, k+1 \}\) such that \(e \in N_{i, t^*} \setminus N_{i, t^* + 1}\). We proceed by revealing the randomness and selecting \(j_i\) uniformly at random from the set \([k]\). With the choice of \(j_i\), the next step involves defining \(\widehat{N}_{i + 1} = \sselect{$\widehat{N}_i, (c_{j_i} + c_{j_i + 1})/2$}\). \Cref{lem:sample-chain-decomposition-property} ensures that\footnote{Similarly, should this property not be upheld, both our algorithm and the stochastic process will be terminated immediately. This affects the algorithm's selectability by only a negligible margin.} 
 \[N_{i, j_i + 1} = \select{$\widehat{N}_i, c_{j_i + 1}$} \subseteq \widehat{N}_{i+1} \subseteq \select{$\widehat{N}_i, c_{j_i}$} 
 = N_{i, j_i}.\]

Depending on the specific outcome of \(j_i\), we examine the subsequent scenarios:
\begin{itemize}
    \item \textbf{Case 1:} If \(j_i < t^*\), this indicates that 
\[e \in N_{i, t^*} \subseteq N_{i, j_i + 1} \subseteq  \widehat{N}_{i + 1}.\]
Consequently, \(e\) progresses to a subsequent layer. We then increase \(i\) by 1 to move to layer \(i + 1\).

    \item  \textbf{Case 2:} If \(j_i > t^*\), it implies that 
    \[\widehat{N}_{i + 1} \subseteq N_{i,j_i} \subseteq N_{i, t^* + 1}.\]
    
Observe that \(e \notin N_{i, t^* + 1}\). This observation implies that \(e\) is also not a member of \(\widehat{N}_{i + 1}\) and \(N_{i, j_l}\). Consequently, the probability of \(e\) being  selected in this scenario can be expressed as follows:
\begin{align*}
\begin{split}
\Pr[e \text{ is selected } \mid e \text{ is active}] &\geq 1 - \Pr[e \in \Span(((R(\allocs) \cap \widehat{N}_i) \cup \widehat{N}_{i+1})) \setminus \{e\})]\\
& \geq 1 - \Pr[e \in \Span(((R(\allocs) \cap \widehat{N}_i) \cup N_{i, j_i})) \setminus \{e\})] \\
& \geq \frac{1}{2} - \varepsilon.
\end{split}
\end{align*}
The first inequality is derived from the premise that \(e \in \widehat{N}_i\setminus \widehat{N}_{i + 1}\) while we are employing a greedy strategy on the matroid \(\matroid |_{\widehat{N}_i} / \widehat{N}_{i + 1}\). The subsequent inequality emerges from the inclusion \(\widehat{N}_{i +1}\subseteq N_{i, j_i}\), and the final inequality is in accordance with the definition that 
\[N_{i+1, j_i} = \select{$\widehat{N}_i, c_{j_i}$}.\]
Hence, in this scenario, our algorithm demonstrates effectiveness, thereby concluding the stochastic process.

\item  \textbf{Case 3:} If \(j_i = t^*\), the stochastic process immediately terminates with no guarantee for $e$.
\end{itemize}

Our algorithm fails at element $e$ only when the stochastic process ends in Case 3. As \(j_i\) is drawn uniformly at random from \([k]\), the stochastic process falls into Case 3 with a probability of \(\frac{1}{k}\) in each layer. As there are at most \(\ell\) layers in total, the selectability is given as
\begin{align*}
\begin{split}
    \Pr[e \text{ is selected } \mid e \text{ is active}] \geq \InParentheses{\frac12 - \varepsilon} \cdot \InParentheses{1 - \frac{1}{k}}^\ell \geq \frac{1}{2} - O(\varepsilon).
\end{split}
\end{align*} where the last inequality follows from \Cref{lem:sample-chain-decomposition-property}. We hereby conclude our proof.

%% file: appendix.tex
\appendix

\section{Missing Proofs from Section \ref{sec:prophet}}

\subsection{Tie-Breaking for Distributions with Point Masses}
\label{appendix:point_mass}

Our algorithm for prophet inequality in Section \ref{sec:prophet} can only work with continuous distributions.
For distributions with point masses, we use the following tie-breaking technique from~\citep{RWW20} to make them continuous.

For a distribution $\mathcal{D}$ with point masses, we can define $\mathcal{D}'$ as the bivariate distribution $\mathcal{D}\times U([0,1])$, the product of $\mathcal{D}$ and the uniform distribution over $[0,1]$.
Then for $(x,t),(y,u)$ sampled from $\mathcal{D}'$, define $(x,t) > (y,u)$ if either $x > y$ or $x = y$ and $t>u$.
Observe that the equality holds with zero probability.
Therefore, if we define $F(x,t)=\Prx[(y,u)\sim \mathcal{D}']{(x,t)>(y,u)}$, we gain the desired property that $F$ is continuous and $F(x,t) > F(y,u) \Leftrightarrow (x,t) > (y,u)$.

\subsection{Proof of Lemma \ref{lemma:threshold_concentration}}
\label{appendix:threshold_concentration}

\begin{proof}[Proof of Lemma \ref{lemma:threshold_concentration}]

Let the C.D.F. of $\tau_i$ be $\disttaui$. From Observation \ref{obs:is_threshold}, we know that $i \in \OPT$ if and only if $v_i > \tau_i$.\footnote{We neglect the case where $v_i = \tau_i$ since it has zero probability.} Therefore, by definition, $p^i(v) = \Pr[i \in \OPT \mid v_i = v] = \Pr[v_i > \tau_i \mid v_i = v] = \disttaui(v)$. Together with $p_k = \varepsilon(1+\varepsilon)^k-\varepsilon^2$, we have that $p^i(T_i^{(k)}) \in [p_k, p_k + 2\varepsilon]$ if and only if $\disttaui(T_i^{(k)}) \in [\varepsilon(1+\varepsilon)^k-\varepsilon^2, \varepsilon(1+\varepsilon)^k+\varepsilon^2]$.

For each sample $1 \le s \le N$, $i \in U$ and $0 \le k < m$, let $l_{i,k}^s = \I{\disttaui(\tau_i^s) < \varepsilon(1+\varepsilon)^k-\varepsilon^2}$ and $r_{i,k}^s = \I{\disttaui(\tau_i^s) < \varepsilon(1+\varepsilon)^k+\varepsilon^2}$. 
By Chernoff bound we have
\begin{align*}
    \Prx{\sum_{s=1}^N l_{i,k}^s \ge \varepsilon(1+\varepsilon)^kN}&\le e^{-N\varepsilon^4/3} \le \frac{\varepsilon}{2nm},\\
    \Prx{\sum_{s=1}^N r_{i,k}^s \le \varepsilon(1+\varepsilon)^kN}&\le e^{-N\varepsilon^4/2} \le \frac{\varepsilon}{2nm}.
\end{align*}

Here $n$ is the number of elements in $U$, and $N$ is the number of samples. Observe that if $\disttaui(T^{(k)}_{i}) \notin [\varepsilon(1+\varepsilon)^k-\varepsilon^2,\varepsilon(1+\varepsilon)^k+\varepsilon^2]$, we either have $\sum_{s=1}^N l_{i,k}^s \ge \varepsilon(1+\varepsilon)^kN$ or $\sum_{s=1}^N r_{i,k}^s \le \varepsilon(1+\varepsilon)^kN$, thus
\begin{equation*}
    \Prx{\disttaui(T^{(k)}_{i}) \notin [\varepsilon(1+\varepsilon)^k-\varepsilon^2,\varepsilon(1+\varepsilon)^k+\varepsilon^2]} \le \frac{\varepsilon}{nm}.
\end{equation*}
Then by a union bound over all $i \in U$ and $0 \le k < m$, we conclude that
\begin{equation*}
\Prx{\exists i, k \text{ s.t. } \disttaui(T_i^{(k)}) \notin [\varepsilon(1+\varepsilon)^k-\varepsilon^2,\varepsilon(1+\varepsilon)^k+\varepsilon^2]} \le \varepsilon,
\end{equation*}
which implies $\Prx{\forall i, k \text{ s.t. } \disttaui(T_i^{(k)}) \in [\varepsilon(1+\varepsilon)^k-\varepsilon^2,\varepsilon(1+\varepsilon)^k+\varepsilon^2]} \ge 1-\varepsilon$.

Thus $\Prx{\forall i, k \text{ s.t. } p^i(T_i^{(k)})  \in [p_k, p_k + 2\varepsilon]} \ge 1-\varepsilon$.
\end{proof}

\subsection{Proof of \Cref{lemma:monotone_mapping}}
\label{appendix:monotone-exchange}

We will need the strong basis exchange property of matroids.

\begin{lemma}[Strong basis exchange, Theorem 39.12 in \cite{schrijver2003combinatorial}] \label{lemma:strong_exchange}
 Let $\matroid = (U, \indset)$ be a matroid. Let $A$ and $B$ be two bases and $x \in A \setminus B$. Then there exists $y \in B \setminus A$ such that $A - x + y$ and $B - y + x$ are both bases.
\end{lemma}

\paragraph{Proof of Lemma~\ref{lemma:monotone_mapping}} 
Let $U = [n]$ and $w(1) \geq w(2) \geq \cdots \geq w(n)$. Suppose the rank of the matroid is $r$. We denote the elements in $A$ by $\{a_i\}_{i \in [r]}$ where $a_i$'s are in ascending order. Consequently, we have $w(a_1) \ge w(a_2) \ge \cdots \geq w(a_r)$. Without loss of generality, we can assume that the basis $A$ is generated by feeding $1, 2, \dots, n$ to the greedy algorithm sequentially.

Next, we state our algorithm for constructing the bijection $f$. We will need the strong basis exchange property stated in Lemma~\ref{lemma:strong_exchange}. 

\begin{algorithm}[H]
	\SetAlgoLined
	\caption{Construct mapping $f$}
	\label{alg:mapping}
	{Recall $r$ is the rank of $\matroid$}
	
	{$A_{r + 1} \gets A$}
	
	\For{$i \gets$ $r, r - 1, \dots, 1$ in descending order}{
		\If {$a_i \in A_{i + 1} \setminus B$} {
			{By Lemma \ref{lemma:strong_exchange}, there exists $b_i \in B \setminus A_{i + 1}$ such that both $A_{i + 1} - a_i + b_i$ and $B - b_i + a_i$ are bases}
			
			{Let $f(a_i) \gets b_i$}
			
			{$A_i \gets A_{i + 1} - a_i + f(a_i)$}
		}
		\Else {
			{Let $f(a_i) \gets a_i$}
			
			{$A_i \gets A_{i + 1}$}
		}
	}
\end{algorithm}

We first prove the following facts about Algorithm \ref{alg:mapping}. 

\begin{fact} \label{fact:subset}
	$\{f(a_i), f(a_{i + 1}), \dots, f(a_r)\} \subset A_i$. 
\end{fact} 
\begin{proof}
	For all $j \geq i$, if $a_j \in A_{j + 1} \setminus B$, we have $f(a_j) \in B \setminus A_{j + 1}$ which implies $f(a_j) \not\in \{a_i, \dots, a_j\}$. Together with $f(a_j) \in A_j$, we can see that $f(a_j)$ is never eliminated and $f(a_j) \in A_i$. If $a_j \not\in A_{j + 1} \setminus B$, we know $f(a_j) = a_j$, and it is in both $A_j$ and $A_i$. 
\end{proof}

\begin{fact} \label{fact:bijection}
	The resulting mapping $f$ is a bijection. 
\end{fact} 
\begin{proof}
	If $a_i \in A_{i + 1} \setminus B$, since $b_i \in B \setminus A_{i + 1}$, we know $f(a_i) = b_i \not\in \{f(a_{i + 1}), \dots, f(a_r)\}$ (by Fact \ref{fact:subset}). If $a_i \not\in A_{i + 1} \setminus B$, $f(a_i)$ equals $a_i$ itself. Note for all $j > i$, $A_{j + 1} - a_j + f(a_j)$ is a basis and $a_i \in A_{j + 1}$. Therefore $f(a_i) = a_i$ cannot equal to any of such $f(a_j)$. 
	
	Thus $f(a_i) \not= f(a_j)$ for all $j > i$. $f$ is an injection. Since $A$ and $B$ are two bases and $|A| = |B|$, $f$ is a bijection. 
\end{proof}

\begin{fact} \label{fact:monotone}
	For all $a_i \in A$, $w(f(a_i)) \leq w(a_i)$. 
\end{fact}
\begin{proof}
	If $w(f(a_i)) > w(a_i)$, we know $f(a_i) < a_i$. Thus $f(a_i)$ is fed to the greedy algorithm before $a_i$. 
	On the other hand, since $f(a_i) \ne a_i$, we must have $a_i \in A_{i + 1} \setminus B$ and $A_{i + 1} - a_i + f(a_i)$ is a basis. Therefore $\{a_1, a_2, \dots, a_{i - 1}, f(a_i)\} \in \indset$. Then our greedy algorithm should have selected $f(a_i)$ before it selected $a_i$. Since we enumerate $a_i$ in decreasing order of their index, if our algorithm selected $f(a_i)$, we must have $f(a_i) \in A_{i + 1}$. This contradicts with the fact $f(a_i) \in B \setminus A_{i + 1}$. 
\end{proof}

Finally, we verify that for all $a_i$, $B - f(a_i) + a_i$ is a basis. It is trivial when $f(a_i) = a_i$. Otherwise, $B - f(a_i) + a_i$ is also a basis according to our construction. This finishes the proof of the lemma.
\qed

\section{Missing Proofs from Section \ref{sec:orcs}}

\subsection{Uniqueness of $\mathsf{Select}(N,c)$ and $\widehat{\mathsf{Select}}(N,c)$}
\label{appendix:ocrs-protection-uniqueness}

In this section, we show $\select(N,c)$ can be uniquely defined by the set $S$ returned by Algorithm~\ref{alg:protection}, no matter which element $e$ is chosen each time at Line~\ref{line:choose_element}.

Let $\mathcal{T}$ be the family of sets that do not satisfy the condition in Line~\ref{line:choose_element}, i.e., for every $T \in \mathcal{T}$, $e \in N \setminus T$,
\begin{equation*}
\Prx{e \in \Span(((R(\allocs) \cap N) \cup T) \setminus \{e\})} \le c.
\end{equation*}
By definition, the set $S$ returned by Algorithm~\ref{alg:protection} always belongs to $\mathcal{T}$.
We claim such $S$ is the unique minimum set in $\mathcal{T}$.
The plan is to prove by induction that the changing $S$ in always remains a subset of $T$ during the execution of Algorithm~\ref{alg:protection} for every $T \in \mathcal{T}$.
Initially, $S = \emptyset$ and the statement trivially holds.
Each time an element is added to $S$, we must have
$$\Prx{e \in \Span(((R(\allocs) \cap N) \cup T) \setminus \{e\})} \ge \Prx{e \in \Span(((R(\allocs) \cap N) \cup S) \setminus \{e\})} > c,$$
where the first inequality is due to the monotonicity of $\Span$ and the induction hypothesis $S \subseteq T$.
That implies $e \in T$ by definition of the family $\mathcal{T}$. 
Therefore, $S \cup \{e\} \subseteq T$ and we have completed the proof.
The same argument holds if we replace $\mathbf{Pr}$ with $\widehat{\mathbf{Pr}}$, and thus $\sselect(N,c)$ is also uniquely defined by Algorithm~\ref{alg:sample-select} (for any fixed samples $R_1,R_2,\dots,R_s$).

\subsection{Proof of Lemma~\ref{lemma:OCRS:MatroidTerminate}}
\label{appendix:ocrs-termination}

\begin{proof}[Proof of Lemma~\ref{lemma:OCRS:MatroidTerminate}]
    Let $r$ be $\rank(S)$. Take $e_1, e_2, \dots, e_r$ from $S$ as follows:  $e_i$ is the item that increases the rank of $S$ from $i - 1$ to $i$ during the algorithm. 
    
    Before the formal proof, we first explain the intuition. Note we know that $\Prx{e_i \notin \Span(((R(\allocs) \cap N)\cup \{e_1,e_2,\dots,e_{i-1}\}) \setminus \{e_i\})} < 1 - c$ holds for all $i\in [r]$ since $e_i$ is added to $S$ by our algorithm. So instead of directly consider $\rank(S)$, we start from the random set $R(\allocs) \cap N$ and add $e_1, e_2, \dots, e_r$ to it in order. Each item we add $e_i$, the probability that it increase the rank by one is strictly smaller than $1 - c$. Finally, we get a set with rank $\Ex{\rank(S \cup (R(\allocs) \cap N))}$ in expectation. This is an upper bound of $\rank(S)$, and we can further upper bound the rank of it by $\rank(N)$ using the fact that $e_i$ increase its rank with probability strictly less than $1 - c$. 

    Formally, we have
\begin{align*}
    \rank(S)&\le \Ex{\rank(S\cup (R(\allocs) \cap N))}\\
            &\leq \Ex{|R(\allocs) \cap N|} + \sum_{i=1}^r \Prx{e_i \notin \Span(((R(\allocs) \cap N)\cup \{e_1,e_2,\dots,e_{i-1}\}) \setminus \{e_i\})}\\
            &< b \rank(N) + (1-c)\rank(S),
\end{align*}
where we used the fact that $\Ex{|R(\allocs) \cap N|} \leq \sum_{i \in N} \alloci \leq b \rank(N)$ and $\Prx{e_i \notin \Span(((R(\allocs) \cap N)\cup \{e_1,\cdots,e_{i-1}\}) \setminus \{e_i\})} < 1 - c$. 
As a result, $\rank(S) < \frac{b}{c} \cdot \rank(N)$.
\end{proof}

\section{Non-existence of $\Omega(1)$-Balanced Matroid CRS with $o(\log n)$ Samples}
\label{appendix:ocrs-sample-lower-bound}

A previous work~\citep{HPT22} showed the non-existence of $\Omega(1)$-balanced matroid contention resolution scheme (CRS) with $O(1)$ samples.
We note that their impossibility result can be extended to any $\Omega(1)$-balanced matroid CRS with $o(\log n)$ samples by using the same hard instance.
As a direct consequence, it is also impossible to obtain a constant-selectable OCRS with $o(\log n)$ samples.

One of their hard instance is the graphic matroid $\mathcal{M}=(E,\mathcal{I})$ defined on the complete bipartite graph $K_{N,M}=(U \cup V, E)$ with bipartition $U=\{u_1,u_2,\dots,u_N\}$ and $V=\{v_1,v_2,\dots,v_M\}$. 
Let $\polytope_{\mathcal{M}}$ be the polytope of the matroid and $\delta(u)$ be the set of edges incident to a vertex $u \in U \cup V$.

\input{graphical_fig}

Consider $N$ points $\allocs^1, \allocs^2, \dots, \allocs^N$ where for every $i \in [N]$ and $e \in E$, 
\begin{equation*}
\alloc^i_e = \begin{cases}
1 & \text{if } e \in \delta(u_i), \\ 
\frac{1}{M} & \text{otherwise}. 
\end{cases}
\end{equation*}

Their intermediate result can be summarized as follows:
\begin{lemma}[Section 3.1,~\citep{HPT22}]
    The points $\allocs^1, \allocs^2, \dots, \allocs^N$ are in the matroid polytope $\polytope_{\mathcal{M}}$. 
    Moreover, for any $N,M \ge 1$, $s \ge 0$, and $0 < c \le 1$, if there exists a $c$-balanced CRS with $s$ samples for $\allocs^1, \allocs^2, \dots, \allocs^N \in \polytope_{\mathcal{M}}$, then
    $$c \le \frac{1}{M} + \frac{M^{(s+1)M}}{N}.$$
\end{lemma}
Then they concluded the non-existence of $\Omega(1)$-balanced CRS for any constant $s$ by letting $N \gg M^M$.
The same argument, notably, remains valid to cases where $s$ is beyond constant.
Specifically, for any $s \le \frac{\log N}{2M \log M} - 1$, the inequality above becomes
$$c \le \frac{1}{M} + \frac{1}{\sqrt{N}}.$$
In other words, a $(\frac{1}{M} + \frac{1}{\sqrt{N}})$-balanced CRS for $\mathcal{M}$ requires at least $\frac{\log N}{2M \log M}$ samples. By choosing $N \gg M^M$ (and note that $n=|E|=NM$), we conclude that $\Omega(1)$-balanced matroid CRS is impossible given $o(\log n)$ samples.

%% file: graphical_fig.tex
\begin{figure}[H]
\centering
\begin{subfigure}[b]{0.45\textwidth}
\centering
\begin{tikzpicture}[thick,amat/.style={matrix of nodes,nodes in empty cells,
  fsnode/.style={draw,solid,circle,execute at begin node={$u_{\the\pgfmatrixcurrentrow}$}},
  ssnode/.style={draw,solid,circle,execute at begin node={$v_{\the\pgfmatrixcurrentrow}$}}}]

 \matrix[amat,nodes=fsnode,label=above:$U$,row sep=1em,dashed,draw,rounded corners] (mat1) {
 \\
 \\ 
 \\
 \\
 \\};

 \matrix[amat,right=2cm of mat1,nodes=ssnode,label=above:$V$,row sep=1em,dashed,draw,rounded corners] (mat2) {\\
 \\ 
 \\};
\foreach \x in {1,...,5}
    \foreach \y in {1,...,3} 
    {\draw  (mat1-\x-1) edge[thick, dotted] (mat2-\y-1); } 

\foreach \y in {1,...,3} 
	{\draw  (mat1-4-1) edge[very thick] (mat2-\y-1); }
\end{tikzpicture}
\caption{The complete bipartite graph $K_{N,M}$ and $\allocs^i$. Here $i = 4$, edges adjacent to $u_i$ has probability $x_e^i = 1$ of being active, while other edges each only has probability $\frac{1}{M}$ of being active. In any $\Omega(1)$-balanced CRS, they have to be selected with constant probability.}
\end{subfigure}
\hfill
\begin{subfigure}[b]{0.45\textwidth}
\centering
\begin{tikzpicture}[thick,amat/.style={matrix of nodes,nodes in empty cells,
  fsnode/.style={draw,solid,circle,execute at begin node={$u_{\the\pgfmatrixcurrentrow}$}},
  ssnode/.style={draw,solid,circle,execute at begin node={$v_{\the\pgfmatrixcurrentrow}$}}}]

 \matrix[amat,nodes=fsnode,row sep=1em] (mat1) {
 \\
 \\ 
 \\
 \\
 \\};

 \matrix[amat,right=2cm of mat1,nodes=ssnode,label=above:$V$,row sep=1em,dashed,draw,rounded corners] (mat2) {\\
 \\ 
 \\};
\draw  (mat1-1-1) edge (mat2-1-1); 
\draw  (mat1-5-1) edge (mat2-3-1); 

\draw[thick,dashed,rounded corners] ($(mat1-3-1)+(-0.55,0.6)$)  rectangle node[left, xshift = -0.5cm] {$U^*$} ($(mat1-4-1)+(0.55,-0.6)$);

\foreach \y in {1,...,3} 
	{\draw  (mat1-3-1) edge (mat2-\y-1); \draw  (mat1-4-1) edge (mat2-\y-1); }
\end{tikzpicture}
\caption{A realization $R(\allocs^i)$ of this instance. Let $U^*$ be the set of all vertices on left side of degree $M$. If $N$ is large enough, there will be many vertices happen to be in $U^*$. These vertices in $U^*$ are indistinguishable to CRS, and $u_i$ ($i = 4$) is hidden between them.}
\end{subfigure}
\caption{The hard instance for graphic matroids in~\citep{HPT22}}
\label{fig:Graphical}
\end{figure}
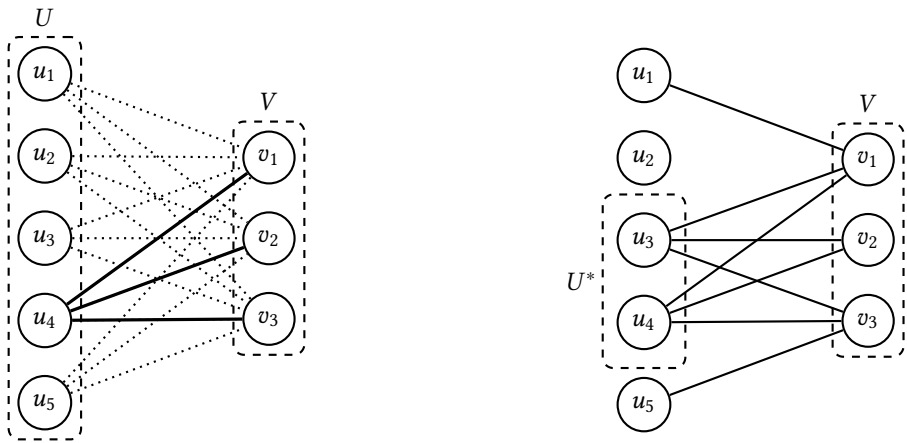